\documentclass[a4paper,onecolumn,allowtoday]{quantumarticle}

\pdfoutput=1

\usepackage[utf8]{inputenc}
\usepackage[T1]{fontenc}
\usepackage{amsmath}
\usepackage{amssymb}
\usepackage{amsthm}
\usepackage{mathtools}
\usepackage{url}
\usepackage{booktabs}
\usepackage{enumitem}
\usepackage{xcolor}
\usepackage{graphicx}

\usepackage{appendix}
\usepackage[numbers,sort&compress]{natbib}
\usepackage{hyperref}

\hypersetup{
  colorlinks=true,
  linkcolor=blue!50!black,
  citecolor=blue!50!black,
  urlcolor=blue!50!black,
}

\theoremstyle{plain}
\newtheorem{theorem}{Theorem}[section]
\newtheorem{proposition}[theorem]{Proposition}
\newtheorem{lemma}[theorem]{Lemma}
\newtheorem{corollary}[theorem]{Corollary}
\newtheorem{conjecture}[theorem]{Conjecture}

\theoremstyle{definition}

\newtheorem{remark}[theorem]{Remark}

\newcounter{routealg}

\newcommand{\rt}{\mathrm{rt}}

\newcommand{\spec}{\mathrm{spec}}
\newcommand{\HGP}{\mathrm{HGP}}
\newcommand{\LP}{\mathrm{LP}}

\newcommand{\Tval}{T_{\mathrm{Valiant}}}
\newcommand{\Llay}{L_{\mathrm{layers}}}
\newcommand{\bbase}{\beta_{\mathrm{base}}}
\newcommand{\bhgp}{\beta_{\HGP}}

\newcommand{\bunion}{\beta_{\mathrm{union}}}
\newcommand{\Dbase}{D_{\mathrm{base}}}
\newcommand{\DT}{D_T}
\newcommand{\GT}{G_T}
\newcommand{\Gbase}{G_{\mathrm{base}}}

\newcommand{\Pb}{\mathbb{P}}

\title{Using Tanner Spectral Reduction to Improve Multi-Layer Optical Lattice Routing for Hypergraph-Product and Bivariate Bicycle qLDPC Codes}

\author{Joshua M. Courtney}
\affiliation{University of Georgia, Department of Physics and Astronomy}
\orcid{0009-0002-1631-4277}
\email{Joshua.Courtney1@uga.edu}

\date{July 6, 2026}

\begin{document}
\maketitle

\begin{abstract}
We characterize the Tanner graph spectrum of hypergraph-product (HGP) / lifted-product (LP) codes and bivariate-bicycle (BB) codes, informing qubit routing for three-dimensional reconfigurable qubit architectures. Syndrome-extraction routing depth on HGP/LP Tanner graphs reduces to a single SVD on the base parity-check matrix, using a spectral ratio $\beta_\text{HGP} = (1 + \beta_\text{base})/2$ where $\beta_\text{base} = \sigma_2(H)/\sigma_1(H)$ for the base parity-check matrix, and a diameter identity $D_T = 2 D_\text{base}$ where $D_\text{base}$ is the base Tanner graph diameter. Fourier spectral reduction reveals that the BB Tanner graph spectrum equals the union, over the $l \times m$ grid of characters of $\mathbb{Z}_l \times \mathbb{Z}_m$, of the singular values of a single $2 \times 2$ symbol matrix built from the two defining polynomials. This reduces spectral analysis from an $O((lm)^3)$ diagonalization of the $4lm$-node Tanner graph to $lm$ independent $2 \times 2$ SVDs. These results compose into a multi-layer three-dimensional AOL routing protocol with one-time setup cost $T_\text{Valiant} = O(\log N)$ atom rearrangements amortizable over a memory experiment of $R$ rounds. For a Tanner graph chromatic index $\chi'$ and $L_\text{layers}$ stacked AOL planes, the per-syndrome-cycle depth is $\lceil \chi'/L_\text{layers} \rceil$ AOL pattern activations with no atom motion, an $8\times$ step-count reduction at $L_\text{layers} \geq \chi' = 8$. Contingent on multi-layer AOL hardware, this yields an estimated $\sim50\text{--}300\times$ per-cycle wall-clock advantage over a single-layer AOD baseline (degrading to $\sim5\text{--}100\times$ under AOD-crosstalk overhead), reducing to equality in the single-layer limit. This paper therefore presents a route toward practical routing improvement for future quantum hardware incorporating multi-layer reconfigurable qubit architectures.
\end{abstract}

\section{Introduction}\label{sec:intro}

Recent demonstrations of quantum low-density parity check (qLDPC) code architectures~\cite{bravyi2024high,xu2024constant,breuckmann2021quantum} substantially reduce physical-qubit overhead for fault-tolerant quantum computation compared to surface codes, realizing a constant overhead~\cite{gottesman2013fault}.
These works use routing primitives moving qubits between data and ancilla positions across each syndrome extraction cycle.
While syndrome-extraction circuit depth (number of CNOT layers) is well-understood for these codes~\cite{tremblay2022constant,delfosse2021bounds}, atom-rearrangement depth required to support each circuit layer has yet to be characterized in terms of the underlying Tanner graph spectrum.
Impetus for this characterization stems from both hardware runtime bottlenecks in atom/ion reconfigurations and assessing the potential advantage of incorporating an effective third dimension to reduce routing overhead in neutral atom and trapped ion qubit architectures.

We give a closed-form formula $\bhgp = (1 + \bbase)/2$ and an exact diameter identity $\DT = 2 \Dbase$ for hypergraph-product/lifted-product (HGP/LP) code Tanner graphs (Sec.~\ref{sec:spectral}).
Using Fourier diagonalization (Theorem~\ref{thm:bb-reduction}), we reduce the bivariate bicycle (BB) Tanner graph spectrum to $lm$ independent $2 \times 2$ singular-value decompositions, and prove that no scalar $(1+\bbase)/2$-style one-liner holds for BB codes.
We give a syndrome-extraction protocol and depth bound for a 2D acousto-optic deflector (AOD) atom array augmented with $\Llay$ stacked 3D acousto-optic lens (AOL) planes, applicable uniformly to HGP, LP, and BB codes (Sec.~\ref{sec:protocol}).
Section~\ref{sec:numerical} gives a numerical comparison to Xu et al.~\cite{xu2024constant}'s scheme on the published HGP code family, under matched amortization assumptions, with BB codes (Sec.~\ref{sec:bb-measurement}) measured in the same framework.

We connect to two recent qLDPC architectures, situated among a growing body of neutral-atom qLDPC layout and routing work~\cite{pecorari2025high,poole2024architecture,constantinides2024optimal,zhou2025resource}.
Xu et al.~\cite{xu2024constant} implement HGP/LP codes on reconfigurable atom arrays with 2D divide-and-conquer scrambling.
Bravyi et al.~\cite{bravyi2024high} demonstrate bivariate bicycle (BB) codes on superconducting hardware whose Tanner graph decomposes into two edge-disjoint planar subgraphs (Sec.~\ref{sec:discussion}).
We refer to and pipeline Algorithm 3 from Xu et al.~\cite{xu2024constant}, requiring $2\delta_c = 8$ atom rearrangements per syndrome cycle on single-layer AOD hardware (constant in $N$). 
Our scheme on multi-layer 3D AOL with $\Llay$ stacked planes requires $\lceil 2\delta_c/\Llay \rceil$ AOL pattern activations per cycle (also constant in $N$), with no atom motion in the application phase. 
The one-time setup cost $\Tval$ to bring atoms into canonical positions grows as $\Theta(\log N)$ but amortizes over the $R \gtrsim 10^3$ syndrome rounds of a memory experiment. 
The per-cycle speedup is $\Llay$ for $\Llay \leq \chi' = 8$ and saturates at $8\times$ for $\Llay \geq 8$.
Both schemes are $O(1)$ per cycle in atom-array reconfigurations, and the multi-layer architecture contributes a constant-factor advantage.
For non-QC HGP/LP codes, both schemes remain applicable, since Xu et al.~\cite{xu2024constant}'s Algorithm 1 handles arbitrary 1D permutations as well as shifts, while our AOL patterns are arbitrary matchings rather than shifts, but remain pre-storable on Bluvstein~\cite{bluvstein2024logical}-style hardware.
With K\"onig-optimal edge coloring, $\chi'(\GT) = 2\delta_c$ for both QC and non-QC bases. 
The only non-QC cost becomes an offline preprocessing burden of computing and storing $\chi'$ AOL patterns (Sec.~\ref{sec:nonqc-results}). Table~\ref{tab:regimes} situates the multi-layer AOL protocol against existing routing regimes for syndrome extraction at qLDPC scale.

\begin{table}[ht]
\centering
\small
\begin{tabular}{p{0.2\linewidth}p{0.2\linewidth}p{0.1\linewidth}p{0.15\linewidth}p{0.2\linewidth}}
\toprule
Routing scheme & Per-cycle cost & Setup cost & Applicable codes & Hardware demonstrated \\
\midrule
Grid + nearest-neighbor shuttling & $\Theta(\sqrt{N})$ atom motions & --- & surface code & yes \cite{bluvstein2024logical} \\
Single-layer AOD, Xu et al. Alg.~3 \cite{xu2024constant} & $2\delta_c$ atom motions & --- & HGP / LP & yes \cite{bluvstein2024logical,xu2024constant} \\
Two edge-disjoint planar overlays \cite{bravyi2024high} & $\chi'(\GT^{\mathrm{BB}}) = 6$ couplings & --- & Bravyi BB codes only & yes (superconducting) \cite{bravyi2024high} \\
Multi-layer 3D AOL, this work, $\Llay \geq \chi'$ & $1$ AOL pattern activation & $\Tval = O(\log N)$ & HGP / LP / BB & partial ($\Llay \lesssim 2$ today \cite{bluvstein2024logical}) \\
\bottomrule
\end{tabular}
\caption{Four-regime comparison of syndrome-extraction routing schemes.
Per-cycle costs are constant in $N$ for all qLDPC schemes (surface-code grid scales as $\Theta(\sqrt N)$). 
Our scheme adds a one-time $\Tval = O(\log N)$ setup cost that amortizes to negligible per cycle within the first few rounds of a memory experiment (Theorem~\ref{thm:amortization}).
The setup-cost constant $k_{\mathrm{emp}} = \Tval / \log_2 N$ is empirically $\approx 0.55$ across the validated range $N \leq 100{,}000$ (Table~\ref{tab:k-emp}).
BB codes \cite{bravyi2024high} achieve a lower per-cycle constant ($\chi'^{\mathrm{BB}} = 6$ vs.\ $\chi'^{\mathrm{HGP},(3,4)} = 8$) by construction, as reflected in the saturation point (Sec.~\ref{sec:bb-measurement}).}
\label{tab:regimes}
\end{table}

\section{Background}\label{sec:background}

\subsection{Ramanujan graphs and routing}

Following previous work~\cite{courtney2026permutation}, a $d$-regular graph $G$ is Ramanujan~\cite{lubotzky1988ramanujan,marcus2018interlacing,margulis1988explicit} if its non-trivial eigenvalues satisfy $|\lambda_i| \leq 2\sqrt{d-1}$ for $i \neq 1, N$.
Random regular graphs are nearly Ramanujan~\cite{friedman2008proof} (see also the expander survey~\cite{hoory2006expander}).
The spectral ratio $\beta = \lambda_\star / d$  where $\lambda_\star = \max(|\lambda_2|, |\lambda_N|)$ controls the routing depth via Valiant's two-phase scheme~\cite{valiant1982scheme,valiant1981universal}.
There we find that for a $d'$-regular Ramanujan graph $G$ on $N \geq 16$ vertices with spectral ratio $\beta$:
\begin{equation}\label{eq:papier1-bound}
\rt(G) \leq \frac{4(d' + 6)}{d' \log_2(1/\beta)} \log_2 N + 19 \log_2 N,
\end{equation}
acting as a tight bound in the worst case but loose by $\sim 25\times$ for HGP code Tanner graphs (see Sec.~\ref{sec:tighter-bound}).

\subsection{HGP and LP qLDPC code Tanner graphs}

For a classical parity-check matrix $H$ (size $m \times n$) with bipartite Tanner graph $\Gbase$, the HGP code $\HGP[H, H']$ (Tillich--Z\'emor~\cite{tillich2013quantum}, generalizing the quantum hypergraph-product construction of Kovalev--Pryadko~\cite{kovalev2013quantum}) has $N = nn' + mm'$ data qubits, $m_X = mn'$ X-checks, and $m_Z = nm'$ Z-checks.
We focus on the self-product case $H = H'$ throughout this paper.

The full Tanner graph $\GT$ is bipartite, with qubits on one side and all checks on the other. 
We index nodes as
\begin{align*}
\text{L-qubits } (v, w) &\in V \times V, \quad \text{R-qubits } (c, c') \in C \times C, \\
\text{X-checks } (c, w) &\in C \times V, \quad \text{Z-checks } (v, c') \in V \times C,
\end{align*}
where $V = \{1, \ldots, n\}$ are variable nodes and $C = \{1, \ldots, m\}$ are check nodes of $\Gbase$.
Edges follow the standard HGP definition.

LP codes~\cite{panteleev2021quantum} are HGP codes built from a $\mathbb{Z}_L$ voltage cover of a base $H$.
The lifted-product and balanced-product constructions~\cite{panteleev2021quantum,breuckmann2021balanced,panteleev2022asymptotically} and quantum Tanner/expander codes~\cite{leverrier2022quantum,leverrier2015quantum,dinur2023good} are the route to asymptotically good qLDPC codes.
The lifted matrix $H_{\text{lift}}$ has size $mL \times nL$, and the LP Tanner graph is $\HGP[H_{\text{lift}}, H_{\text{lift}}]$.

\subsection{Atom-array hardware: 2D AOD and 3D AOL}

We take the hardware model of a 2D atom array on an $L \times L$ grid ($N = L^2$ atoms), with acousto-optic deflectors (AOD) providing arbitrary atom-pair connectivity through coherent transport~\cite{bluvstein2024logical, xu2024constant}, building on the reconfigurable Rydberg-array platform~\cite{saffman2010quantum,browaeys2020many,bernien2017probing,ebadi2021quantum,scholl2021quantum,barredo2016atom}.
3D acousto-optic lattices (AOL) augment this with multiple stacked ``layers'' of independent AOD patterns, denoted here as $\Llay$, a premise supported by demonstrated 3D atom-array assembly~\cite{barredo2018synthetic} and by recent 3D acousto-optic transport hardware~\cite{lu2026astigmatism,guo2025acousto,picard2025three}.

In previous work (Xu et al.~\cite{xu2024constant}), atom rearrangement uses a divide-and-conquer 1D scrambling algorithm achieving arbitrary permutations in $O(\log L)$ recursive levels.
Each level requires parallel atom motion of distance up to $L = \sqrt{N}$, with per-level wall-clock $\sim 3$~ms at $N \sim 10^4$ scaling as $O(N^{1/4})$ under constant-acceleration physics.
This $\sim 3$~ms figure is a \emph{long-range} cubic-spline transport time (Xu et al. Methods Eq.~7, a $\sim\!500~\mu$m sweep; cf.\ the $3$~ms AOD rearrangement sweep of Endres et al.~\cite{endres2016atom}), taken as the cost for the one-time full-array scrambling of the setup phase. We distinguish this transport time from the \emph{short-range} per-gate move used inside a syndrome cycle, whose demonstrated characteristic time is $\sim 200~\mu$s ($0.55~\mu$m\,$\mu$s$^{-1}$ cubic-velocity profile)~\cite{bluvstein2022quantum,bluvstein2024logical}.

\subsection{Comparison baseline}

Per Xu et al.~\cite{xu2024constant} Algorithm 3, pipelining requires $2\delta_c$ rearrangement layers per syndrome cycle, where $\delta_c$ is the chromatic index of the underlying classical parity-check Tanner graph.
For a $(3,4)$-biregular base, $\delta_c = 4$, giving:
\begin{equation}\label{eq:xu-baseline}
T^{\text{Xu et al.}}_{\text{cycle}} = 2 \delta_c = 8,
\end{equation}
constant in $N$ for fixed code structure.
Each per-cycle rearrangement layer is a short-range data$\leftrightarrow$ancilla move, whose demonstrated characteristic wall-clock is $\sim 200~\mu$s~\cite{bluvstein2022quantum,bluvstein2024logical}, giving $8 \times 200~\mu$s $\approx 1.6$~ms per syndrome cycle at $N \sim 10^4$. Gate-depth count of 16 entangling layers (``$4\delta_c$'' in their notation) is a separate quantity from the rearrangement-layer count.

\section{Spectral Characterization of qLDPC Tanner Graphs}\label{sec:spectral}

This section presents structural results for the Tanner graphs of two qLDPC families.

\subsection{Spectrum decomposition (Result 1a)}\label{sec:spec-decomp}

\begin{proposition}\label{prop:spec-decomp}
For an HGP code $\HGP[H, H]$ with rank $r$ base parity-check matrix $H$, every eigenvalue of the full Tanner graph adjacency matrix $A_{\GT}$ lies in the multiset
\begin{equation}\label{eq:spec-multiset}
\mathrm{spec}(A_{\GT}) \subseteq \{\pm(\sigma_i \pm \sigma_j) : 1 \leq i, j \leq r\} \cup \{\pm \sigma_k : 1 \leq k \leq r\} \cup \{0\}
\end{equation}
where $\sigma_1 \geq \cdots \geq \sigma_r > 0$ are the nonzero singular values of $H$.
We refer to the first set as \textbf{product modes} and to the second set as \textbf{boundary modes}.
\end{proposition}

\begin{proof}
The Tanner graph is bipartite between qubits and checks, so its adjacency is
$A_{\GT} = \left(\begin{smallmatrix} 0 & B \\ B^{T} & 0 \end{smallmatrix}\right)$
with $B$ the qubit-to-check biadjacency, and $\mathrm{spec}(A_{\GT}) = \{\pm \sigma : \sigma \in \mathrm{sv}(B)\} \cup \{0\}$.
Using the HGP stabilizer blocks $H_X = (H \otimes I_n \mid I_m \otimes H^{T})$ and $H_Z = (I_n \otimes H \mid H^{T}\otimes I_m)$ (Tillich--Z\'emor~\cite{tillich2013quantum}), $B = (H_X^{T}\mid H_Z^{T})$ decomposes over the two qubit sectors ($\mathcal{Q}_0 = V\times V$, $\mathcal{Q}_1 = C \times C$), and a direct computation (Appendix~\ref{app:blockSVD}) gives
\[
B B^{T} = \begin{pmatrix} (H^{T}H)\otimes I_n + I_n \otimes (H^{T}H) & 2\,H^{T}\otimes H^{T} \\[2pt] 2\,H\otimes H & (HH^{T})\otimes I_m + I_m \otimes (HH^{T}) \end{pmatrix}.
\]
Fix the SVD $H = \sum_k \sigma_k\, u_k w_k^{T}$ ($Hw_k = \sigma_k u_k$, $H^{T}u_k = \sigma_k w_k$, $1\le k\le r$). For each ordered pair $(i,j)$ with $\sigma_i,\sigma_j>0$, the two-dimensional space $\mathrm{span}\{\,w_i\otimes w_j,\; u_i\otimes u_j\,\}$ is $BB^{T}$-invariant, and in this basis
\[
BB^{T}\big|_{(i,j)} = \begin{pmatrix} \sigma_i^2 + \sigma_j^2 & 2\sigma_i\sigma_j \\ 2\sigma_i\sigma_j & \sigma_i^2 + \sigma_j^2 \end{pmatrix},
\qquad \text{eigenvalues } (\sigma_i \pm \sigma_j)^2 ,
\]
so $B$ has singular values $\sigma_i + \sigma_j$ and $|\sigma_i - \sigma_j|$ (the \textbf{product modes}). When one factor is a kernel/cokernel direction of $H$ ($\sigma_j = 0$) the off-diagonal coupling vanishes and $w_i\otimes w_{j_0}$ (resp.\ $u_i \otimes u_{j_0}$) is an eigenvector with eigenvalue $\sigma_i^2$, giving singular value $\sigma_i$ (the \textbf{boundary modes} $\pm\sigma_k$); pairs of kernel/cokernel directions give the zero modes. Collecting these accounts for the full multiset~\eqref{eq:spec-multiset}. The complete arithmetic, including the mixed-product identities $(H^{T}\!\otimes I_n)(I_m\otimes H^{T}) = H^{T}\!\otimes H^{T}$ that produce the factor $2$, is given in Appendix~\ref{app:blockSVD}.
\end{proof}

Across nine random $(3,4)$-biregular bases at $n = 12$ (HGP code on 441 nodes), we account for all eigenvalues by the multiset of Eq.~\eqref{eq:spec-multiset}.
Full-rank $H$ (rank $r = m = 9$) gives 75-84\% product-mode eigenvalues, while rank-deficient $H$ (e.g., $r = 7$ or $8$) gives 56-76\% product-modes with the remainder split between boundary and zero modes.
The mean product fraction over the tested seeds is $\approx 70 \pm 9\%$. 
Figure~\ref{fig:spectrum} visualizes this decomposition for the $[[225, 9, 4]]$ instance, and every measured eigenvalue lands on a predicted product or boundary position from the multiset Eq.~\eqref{eq:spec-multiset}.

\begin{figure}[ht]
\centering
\includegraphics[width=0.95\linewidth]{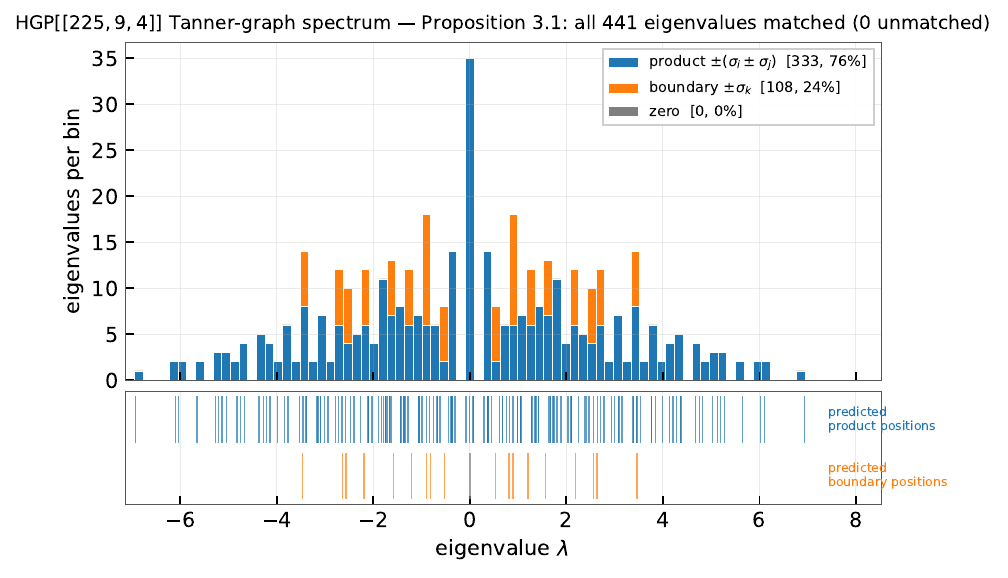}
\caption{Spectrum decomposition of the HGP code $[[225, 9, 4]]$ (seed 1, full-rank base $H$). 
Top: histogram of all 441 Tanner graph eigenvalues, stacked by Proposition~\ref{prop:spec-decomp} classification, as product modes $\pm(\sigma_i \pm \sigma_j)$ (blue, 333 eigenvalues, $76\%$) and boundary modes $\pm \sigma_k$ (orange, 108 eigenvalues, $24\%$). Zero unmatched eigenvalues.
Bottom: the predicted multiset positions as tick marks, aligned with the top-panel histogram.}
\label{fig:spectrum}
\end{figure}

\subsection{Closed-form spectral ratio (Result 1b)}\label{sec:closed-form}

\begin{theorem}\label{thm:closed-form}
For an HGP code $\HGP[H, H]$ whose base parity-check matrix $H$ has a simple top singular value ($\sigma_1(H) > \sigma_2(H)$, i.e.\ $\bbase = \sigma_2(H)/\sigma_1(H) < 1$), the non-trivial bipartite spectral ratio of the Tanner graph $\GT$ is
\begin{equation}\label{eq:closed-form}
\bhgp = \frac{1 + \bbase}{2}.
\end{equation}
Equivalently, the spectral gap is halved: $1 - \bhgp = (1 - \bbase)/2$.
\end{theorem}

\begin{proof}
By hypothesis $\sigma_1(H) > \sigma_2(H)$, so $\bbase < 1$ and the top singular value is simple. 
(The excluded boundary case $\sigma_1 = \sigma_2$ is $\bbase = 1$, where the base graph is not a spectral expander and the routing bound of Eq.~\eqref{eq:papier1-bound} is vacuous and of no interest for routing.)
By Proposition~\ref{prop:spec-decomp} the Tanner spectrum is the multiset of Eq.~\eqref{eq:spec-multiset}.
Its Perron eigenvalue is the product mode with $i=j=1$, namely $2\sigma_1(H)$.
Because $\sigma_1$ is simple, no other pair $(i,j)$ attains $\sigma_i + \sigma_j = 2\sigma_1$, so this Perron value is itself simple.

The non-trivial second eigenvalue (in absolute value, excluding the $\pm 2\sigma_1$ Perron pair) is the maximum over:
\begin{enumerate}
\item Product modes $\pm(\sigma_i + \sigma_j)$ with $(i, j) \neq (1, 1)$:
maximum is $\sigma_1 + \sigma_2$.
\item Product modes $\pm(\sigma_i - \sigma_j)$: maximum is
$\sigma_1 - \sigma_r \leq \sigma_1$.
\item Boundary modes $\pm \sigma_k$: maximum is $\sigma_1$.
\end{enumerate}
For $\sigma_2 > 0$, $\sigma_1 + \sigma_2 > \sigma_1$, so the product
mode dominates. Then
\begin{equation*}
\bhgp = \frac{\sigma_1 + \sigma_2}{2 \sigma_1} = \frac{1 + \sigma_2/\sigma_1}{2} = \frac{1 + \bbase}{2}.
\end{equation*}
For the rank-1 edge case ($\sigma_2 = 0$, i.e., $\beta_{\rm base} = 0$), the second eigenvalue reduces to the boundary mode $\sigma_1$, giving $\beta_{\HGP} = 1/2 = (1 + 0)/2$.
\end{proof}

Numerically, we evidence this result with double-precision matching across 26 instances spanning base sizes $n \in \{12, 16, 20, 24\}$ and LP-style lift orders ($L \in \{2, 3, 4, 5, 6, 8\}$).
For example, on the [[225, 9, 4]] HGP code: $\bbase = 0.7609$ predicts $\bhgp = 0.8804$ to all measured digits, implicating computability of the full Tanner graph spectral structure from a single SVD on the base parity-check matrix. The closed form feeds directly into the routing-depth bound, making the ``single SVD'' claim quantitative.

\begin{corollary}[A base SVD determines the HGP routing depth]\label{cor:routing-depth}
Under the hypotheses of Theorem~\ref{thm:closed-form}, substituting $\bhgp = (1+\bbase)/2$ into the routing-depth bound~\eqref{eq:papier1-bound} gives
\begin{equation}\label{eq:routing-depth-base}
\rt(\GT) \;\leq\; \left(\frac{4(d'+6)}{d'\,\log_2\!\big(\tfrac{2}{1+\bbase}\big)} + 19\right)\log_2 N .
\end{equation}
Hence the routing depth of the $N$-node HGP Tanner graph is determined by the single base ratio $\bbase = \sigma_2(H)/\sigma_1(H)$, i.e.\ by one SVD of the (much smaller) base parity-check matrix $H$, with no diagonalization of $\GT$ itself.
\end{corollary}

\begin{proof}
Immediate from Eq.~\eqref{eq:papier1-bound} with $\beta = \bhgp$ and $\log_2(1/\bhgp) = \log_2\!\big(2/(1+\bbase)\big)$ by Theorem~\ref{thm:closed-form}.
\end{proof}

\begin{remark}[Genericity of the simple-$\sigma_1$ hypothesis]\label{rem:generic}
The hypothesis $\sigma_1(H)>\sigma_2(H)$ is generic rather than restrictive: a repeated top singular value is a non-generic coincidence, and it was simple in every random $(c,r)$-biregular base we tested. It holds throughout the Ramanujan regime $\bbase<1$ that the routing application targets.
\end{remark}

\subsection{Diameter identity (Result 2)}\label{sec:diameter}

\begin{theorem}\label{thm:diameter}
For an HGP code $\HGP[H, H]$ where the classical Tanner graph
$\Gbase$ of $H$ is connected with diameter $\Dbase \geq 1$, the full
Tanner graph $\GT$ has diameter
\begin{equation}\label{eq:diameter}
\DT = 2 \Dbase.
\end{equation}
\end{theorem}

\begin{proof}
\textbf{Upper bound.} Every node $u \in V_{\GT}$ has well-defined projections
\[(\pi_1(u), \pi_2(u)) \in V_{\Gbase} \times V_{\Gbase}.\]
Every edge $(u, u') \in E_{\GT}$ satisfies one of $\pi_i(u) = \pi_i(u')$ for $i \in \{1, 2\}$, with the other coordinate's pair forming an edge in $\Gbase$. 
For any two HGP nodes $u, u'$, take a shortest base-graph path in coordinate 1 (length $\leq \Dbase$) followed by a shortest path in coordinate 2 (length $\leq \Dbase$). 
Each base-edge corresponds to one edge in $\GT$, giving a total $\leq 2 \Dbase$.

\textbf{Lower bound.} We exhibit an L-qubit, R-qubit, or X/Z-check pair $(u_0, u_k)$ with 
\[d_{\Gbase}(\pi_1(u_0), \pi_1(u_k)) = d_{\Gbase}(\pi_2(u_0),\pi_2(u_k)) = \Dbase.\] 
Following case-by-case by the parity of the diameter realizer in the bipartite base graph $\Gbase$:
\begin{itemize}
\item $\Dbase$ realized by V-V pair $(p, q)$. Take $u_0 = (p, p)$,
$u_k = (q, q) \in \mathcal{L}$ (L-qubits). Both projections give $\Dbase$.
\item $\Dbase$ realized by C-C pair $(p, q)$. Take $u_0 = (p, p)$,
$u_k = (q, q) \in \mathcal{R}$ (R-qubits).
\item $\Dbase$ realized by V-C pair $(p, q)$. Take $u_0 = (p, q)
\in \mathcal{Z}$ and $u_k = (q, p) \in \mathcal{X}$. Then $\pi_1(u_0) = p$,
$\pi_1(u_k) = q$, $d_{\Gbase}(p, q) = \Dbase$, and similarly for $\pi_2$.
\end{itemize}
In each case we bound the length of an \emph{arbitrary} $\GT$-path $u_0 = x_0, x_1, \dots, x_\ell = u_k$, which shows no shorter route exists. By the edge structure recalled in the upper bound, every step $x_{t}x_{t+1}$ changes one projection coordinate, by an edge of $\Gbase$: call it \emph{type-1} if it fixes $\pi_1$ (moving $\pi_2$ along a $\Gbase$-edge) and \emph{type-2} if it fixes $\pi_2$.
The coordinate $\pi_2$ changes only on type-1 steps, so the images $\pi_2(x_0), \pi_2(x_1), \dots, \pi_2(x_\ell)$, with consecutive repeats deleted, form a $\Gbase$-walk from $\pi_2(u_0)$ to $\pi_2(u_k)$.
Its length equals the number of type-1 steps, which is therefore $\geq d_{\Gbase}(\pi_2(u_0),\pi_2(u_k)) = \Dbase$.
Symmetrically the number of type-2 steps is $\geq d_{\Gbase}(\pi_1(u_0),\pi_1(u_k)) = \Dbase$. 
Since every step has one type, $\ell \geq 2\Dbase$.
As $u_0, u_k$ were chosen with both projection distances equal to $\Dbase$, this gives $\DT \geq 2\Dbase$. 
Combined with the upper bound, $\DT = 2\Dbase$.
\end{proof}

For random $(c, r)$-biregular base graphs, $\Dbase = O(\log n / \log d_{\text{base}})$ where $N_{\HGP} = \Theta(n^2)$, so 
\[\DT = 2 \Dbase = O\!\left(\log \sqrt{N_{\HGP}}\right)= O(\log N_{\HGP}).\]
The constant in front of $\log N_{\HGP}$ is $1/\log d_{\text{base}}$ via the base graph's expansion, which is tighter than the generic diameter bound's constant $1/\log_2(1/\beta_{\HGP})$, improving the constant on a mutual $O(\log N)$ scaling.

\subsection{Spectral reduction for bivariate-bicycle codes}\label{sec:bb-reduction}

The results above rest on the hypergraph-product structure $H_X = (H \otimes I \mid I \otimes H^T)$. 
As it stands, the results above do not extend to bivariate-bicycle (BB) codes.
Those are built from two commuting circulant polynomials rather than a product of a base code with itself. 
BB codes nonetheless carry a different abelian symmetry, being the translation group $\mathbb{Z}_l \times \mathbb{Z}_m$ of the underlying torus. 
This yields an analogous reduction of the routing-relevant spectrum via Fourier diagonalization.

A BB code~\cite{bravyi2024high} on $n = 2lm$ qubits (and its multivariate generalization~\cite{voss2025multivariate}) is specified by two polynomials
\[A(x,y), B(x,y) \in \mathbb{F}_2[x,y]/(x^l - 1,\, y^m - 1),\]
acting as $lm \times lm$ circulant $0/1$ matrices $A, B$.
Its Calderbank-Shor-Steane (CSS) check matrices are $H_X = (A \mid B)$ and $H_Z = (B^T \mid A^T)$, so
\begin{equation}\label{eq:bb-check}
\mathcal{H} = \begin{pmatrix} A & B \\ B^T & A^T \end{pmatrix},
\end{equation}
as the qubit-to-check biadjacency of the Tanner graph (size $2lm \times 2lm$).
The nonzero eigenvalues of the bipartite Tanner adjacency are $\pm \sigma_i(\mathcal{H})$.

\begin{theorem}[BB Tanner spectral reduction]\label{thm:bb-reduction}
Let $\omega_l = e^{2\pi i/l}$, $\omega_m = e^{2\pi i/m}$, and for each
character $(a,b) \in \mathbb{Z}_l \times \mathbb{Z}_m$ define the
polynomial evaluations
\begin{equation}\label{eq:bb-symbol-eval}
\hat{A}(a,b) = \!\!\sum_{(i,j)\in A}\!\! \omega_l^{ia}\,\omega_m^{jb},
\qquad
\hat{B}(a,b) = \!\!\sum_{(i,j)\in B}\!\! \omega_l^{ia}\,\omega_m^{jb},
\end{equation}
and the $2 \times 2$ Hermitian-coupled symbol matrix
\begin{equation}\label{eq:bb-symbol}
M(a,b) = \begin{pmatrix}
\hat{A}(a,b) & \hat{B}(a,b) \\[2pt]
\overline{\hat{B}(a,b)} & \overline{\hat{A}(a,b)}
\end{pmatrix}.
\end{equation}
Then the full nonzero Tanner graph spectrum of the BB code is
\begin{equation}\label{eq:bb-reduction}
\;\operatorname{spec}(A_{\GT}) \setminus \{0\}
= \bigcup_{(a,b)\in \mathbb{Z}_l \times \mathbb{Z}_m}
\bigl\{\, \pm s_1(a,b),\ \pm s_2(a,b) \,\bigr\},
\end{equation}
where $s_1(a,b) \geq s_2(a,b) \geq 0$ are the two singular values of $M(a,b)$. 
The routing-relevant spectral ratio $\beta_{\mathrm{BB}}$ is obtained from $lm$ independent $2 \times 2$ singular-value decompositions, an $O(lm)$ computation in place of the $O\left((lm)^3\right)$ diagonalization of the $4lm$-node Tanner graph.
\end{theorem}

\begin{proof}
The two-dimensional Fourier transform $F = F_l \otimes F_m$ (unitary, with $F_l$ the $l$-point FFT) diagonalizes all circulants in $\mathbb{F}_2[x,y]/(x^l-1, y^m-1)$ lifted to the reals. 
For the polynomial $A$, 
\[F A F^{*} = \operatorname{diag}_{(a,b)}\hat{A}(a,b),\] 
since $\chi_{a,b}(i,j) = \omega_l^{ia} \omega_m^{jb}$ is a simultaneous eigenvector of all torus translations.
Transposition of a circulant reverses each translation, so
\[F A^{T} F^{*} = \operatorname{diag}_{(a,b)} \overline{\hat{A}(a,b)},\]
and likewise for $B$.
Applying the block-unitary $\operatorname{diag}(F, F)$ to $\mathcal{H}$ of
Eq.~\eqref{eq:bb-check} simultaneously diagonalizes all four blocks. 
After the permutation that groups the two copies of each character $(a,b)$, the result is block-diagonal with the $2 \times 2$ blocks $M(a,b)$ of Eq.~\eqref{eq:bb-symbol}. 
Singular values are invariant under the unitaries $F$ and grouping permutation, so the singular spectrum of $\mathcal{H}$ is the disjoint union of the singular spectra of $M(a,b)$. 
The bipartite Tanner adjacency has eigenvalues $\pm \sigma_i(\mathcal{H})$ together with zeros from the qubit/check count imbalance, giving Eq.~\eqref{eq:bb-reduction}.
\end{proof}

Unlike Theorem~\ref{thm:closed-form}, the BB ratio $\beta_{\mathrm{BB}}$ has no closed form depending only on individual block ratios $\beta_A, \beta_B$. 
The Perron value $\sigma_1 = s_1(0,0) = w_A + w_B$ (where $w_A, w_B$ are the polynomial weights) is fixed by the weights, but the second-largest singular value is attained at a nonzero character $(a,b)$ whose location depends on monomial exponents. 
If we assert a transplanted form \[\beta_{\mathrm{BB}} = (1 + \max(\beta_A, \beta_B))/2,\] 
we overshoot the four published Bravyi codes by $0.12$--$0.17$ and the result fails to hold numerically for non-degenerate codes.

We check Theorem~\ref{thm:bb-reduction} against direct diagonalization of the full Tanner adjacency across 52 BB codes, including the four published Bravyi instances~\cite{bravyi2024high} and 48 random codes spanning polynomial weights $w \in \{2,3,4\}$ (Tanner max-degree $\Delta = 2w \in \{4,6,8\}$) and torus sizes up to $lm = 144$ ($N_{\GT} = 576$). 
The maximum eigenvalue discrepancy over the whole sweep was $3.1 \times 10^{-14}$. 
An independent re-derivation taking the singular values of the explicit check matrix~\eqref{eq:bb-check} (rather than the Tanner adjacency) agreed to $5.6 \times 10^{-14}$ on a fresh sample up to $N_{\GT} = 800$.
For the published Bravyi codes the reduction reproduces the measured ratios of Sec.~\ref{sec:bb-measurement}
($\beta_{\mathrm{BB}} = 0.667, 0.872, 0.828, 0.828$) to all reported digits.

As with the product-code results, the routing-relevant spectrum of a BB code is computable without ever
instantiating the full Tanner graph, making spectral screening of BB polynomial families tractable at scales where full diagonalization is infeasible.

\subsection{Random-voltage Ramanujan fraction for LP codes}\label{sec:ram-frac}

For LP codes built from random voltage assignments on a fixed 3×5 base support mask (mirroring the family in Xu et al.~\cite{xu2024constant}), we measured the fraction of voltage assignments yielding Ramanujan-Tanner graphs (Feng--Li bipartite condition $\sigma_2 \leq \sqrt{c-1} + \sqrt{r-1}$ on the lifted classical Tanner graph) across a sweep of lift orders $L \in \{2, \ldots, 16\}$.

\begin{table}[ht]
\centering
\begin{tabular}{rrcc}
\toprule
$L$ & $N_{\LP}$ & Ramanujan fraction (mean $\pm$ std) & Mean $\beta_{\text{lift}}$ (mean $\pm$ std) \\
\midrule
2  & 136  & $87.3\% \pm 1.1\%$  & $0.824 \pm 0.003$ \\
3  & 306  & $96.9\% \pm 0.9\%$  & $0.839 \pm 0.003$ \\
4  & 544  & $84.3\% \pm 1.6\%$  & $0.866 \pm 0.003$ \\
5  & 850  & $89.9\% \pm 1.1\%$  & $0.867 \pm 0.002$ \\
6  & 1224 & $79.0\% \pm 3.3\%$  & $0.887 \pm 0.003$ \\
8  & 2176 & $74.2\% \pm 1.7\%$  & $0.895 \pm 0.001$ \\
12 & 4896 & $66.8\% \pm 3.1\%$  & $0.910 \pm 0.002$ \\
16 & 8704 & $57.6\% \pm 2.1\%$  & $0.917 \pm 0.002$ \\
\bottomrule
\end{tabular}
\caption{Random-voltage Ramanujan fraction for LP codes. 
Each cell aggregates over 5 master seeds $\times$ 500 voltage samples per seed (2500 trials total per $(L, N_{\LP})$ pair).
Reported uncertainties are the cross-seed standard deviation.}
\label{tab:ram-frac}
\end{table}

This phenomenon is also found in Courtney~\cite{courtney2026permutation} Sec.~ 7's covering-tower analysis on the Fano plane (94\% Ramanujan at $k=2$), supporting random-voltage code search as a practical strategy for producing routing-friendly LP code instances.

\section{Multi-Layer AOL Routing Protocol}\label{sec:protocol}

\subsection{Algorithm}\label{sec:algorithm}

We apply a one-time Valiant routing on the union of $\Llay$ overlay graphs to canonical 2D-grid positions, followed by $R$ syndrome cycles in which each chromatic color of $\GT$ is loaded into one of the stacked AOL planes and fired in parallel. 
The protocol is stated as Algorithm~\ref{alg:multilayer} and illustrated in Figure~\ref{fig:architecture}.

\begin{figure}[ht]
\centering
\includegraphics[width=0.98\linewidth]{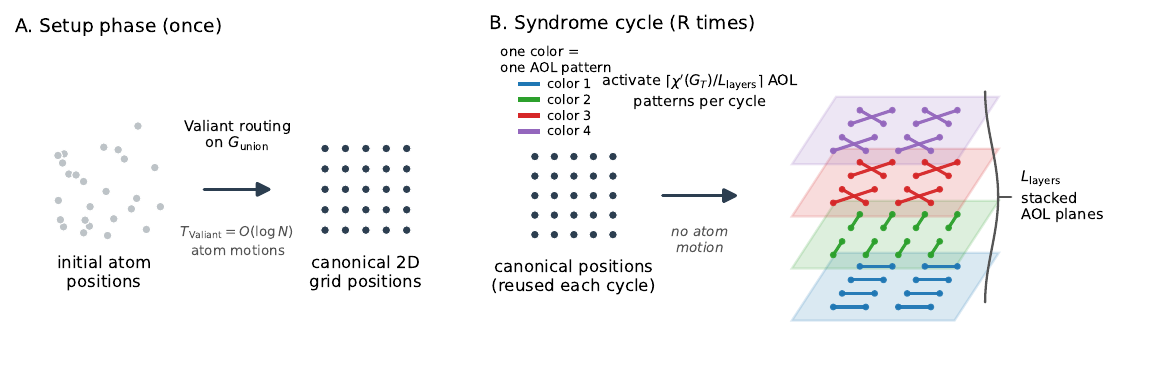}
\caption{Multi-layer 3D AOL routing protocol (Algorithm~\ref{alg:multilayer}, Sec.~\ref{sec:algorithm}).
\textbf{A.}~Setup phase: atoms initialized in arbitrary positions are Valiant-routed on the multi-layer overlay $G_{\mathrm{union}} = \bigcup_{\ell=1}^{\Llay} G_\ell$ to canonical 2D-grid positions supporting Tanner graph $\GT$ chromatic decomposition.
One-time cost $\Tval = O(\log N)$ atom motions (Theorem~\ref{thm:amortization}).
\textbf{B.}~Syndrome cycle (repeated $R$ times): from the canonical positions, each of $\chi'(\GT)$ chromatic colors of a $\GT$ edge coloring is loaded into one of $\Llay$ stacked AOL planes.
The $\Llay$ planes fire simultaneously, completing all $\chi'$ colors in $\lceil \chi'(\GT) / \Llay \rceil$ serial AOL pattern activations.
No atom motion in this phase. For $\Llay \geq \chi'$ the per-cycle cost saturates to a single AOL pattern activation, predicting a $\sim 50$--$300\times$ wall-clock advantage over Xu et al.~\cite{xu2024constant} Algorithm 3 in the canonical-position memory-experiment regime (Sec.~\ref{sec:wallclock-comparison}).}
\label{fig:architecture}
\end{figure}

\begin{figure}[ht]
\fbox{\parbox{\textwidth}{
\refstepcounter{routealg}\label{alg:multilayer}\textbf{Algorithm~\theroutealg.} \textsc{MultiLayerRoute}($\mathcal{C}$, $R$)\\
\textbf{Input:} HGP/LP code $\mathcal{C}$ with Tanner graph $\GT$, number of syndrome rounds $R$, hardware specs $(\Llay, d_0)$.\\
\textbf{Output:} parallel atom-rearrangement schedule for $R$ rounds.

\begin{enumerate}[leftmargin=2em]
\item \textbf{Setup phase} (executed once):
  \begin{enumerate}[leftmargin=1em]
  \item Build $G_{\text{union}} = \cup_{\ell=1}^{\Llay} G_\ell$ where
    each $G_\ell$ is a random $d_0$-regular overlay graph on the data
    atoms.
  \item Apply Valiant--LMR routing~\cite{valiant1982scheme,leighton1994packet,leighton1999fast,upfal1992log} on $G_{\text{union}}$ to permute
    atoms to canonical 2D-grid positions supporting the chromatic
    decomposition of $\GT$. Depth: $\Tval(G_{\text{union}})$ atom
    rearrangements.
  \end{enumerate}

\item \textbf{Syndrome cycle} (repeated $R$ times, no atom motion):
  for each batch of $\Llay$ chromatic colors, activate the $\Llay$
  pre-computed AOL pattern channels and apply CZ gates simultaneously.
  Depth: $\lceil \chi'(\GT) / \Llay \rceil$ AOL pattern activations.
\end{enumerate}

\textbf{Total depth:}
$T_{\text{total}} = \Tval(G_{\text{union}}) + R \cdot \lceil \chi'/\Llay \rceil$.\\
\textbf{Per-cycle depth (amortized):}
$T_{\text{total}}/R = \lceil \chi'/\Llay \rceil + \Tval/R \to \lceil \chi'/\Llay \rceil$.
}}
\end{figure}

\subsection{Amortization theorem}\label{sec:amortization}

\begin{theorem}\label{thm:amortization}
Let $\HGP[H, H]$ be a quasi-cyclic (block-circulant) HGP code with chromatic index $\chi'(\GT)$, on a hardware platform supporting $\Llay$ stacked AOL planes. 
Algorithm~1 has the following cost decomposition for a memory experiment of $R$ syndrome rounds:
\begin{equation}
    \begin{split}
        T_{\text{setup}}      &= \Tval(G_{\text{union}}) \quad \text{(one-time)}, \\
T_{\text{per-cycle}}  &= \lceil \chi'(\GT) / \Llay \rceil \quad \text{(repeated $R$ times)}, \\
\label{eq:amortization}
T_{\text{total}}/R    &= \lceil \chi'(\GT) / \Llay \rceil + \Tval/R.
    \end{split}
\end{equation}

The per-cycle term is constant in $N$. The setup term grows as
$\Theta(\log N)$ but contributes $\Tval/R \to 0$ as $R$ grows.
\end{theorem}

\begin{proof}[Proof sketch]
Setup phase: one Valiant--LMR routing brings atoms to canonical 2D positions supporting the HGP product structure, paid once at the start of the memory experiment.

Per cycle: the chromatic decomposition of $\GT$ is implementable by AOL pulse patterns with no atom motion (atoms remain in canonical positions between cycles). 
For a quasi-cyclic base $H$ this consists of $\chi'$ shift operations on the canonical 2D grid.
With $\Llay$ stacked planes, $\Llay$ different shifts execute in parallel, completing all $\chi'$ colors in $\lceil \chi'/\Llay \rceil$ serial pulse rounds.
The atoms do not need to be re-routed between syndrome rounds because canonical positions are preserved by the ancilla measurement-and-reset cycle.
\end{proof}

The shift-based decomposition is simplest when $H$ is quasi-cyclic.
Each chromatic color of $\GT$ is implementable as a single axis-aligned translation on the 2D embedding, giving a simple shift pattern to each AOL plane.
For non-QC codes (e.g., HGP from random biregular bases), chromatic colors are arbitrary matchings rather than shifts.
By K\"onig's theorem~\cite{konig1916graphen} (see also Vizing~\cite{vizing1964estimate} and~\cite{diestel2024extremal}), $\chi'(\GT) = \Delta(\GT) = 2\delta_c$, being constructible in polynomial time by repeated bipartite matching on the $\Delta$-regular extension~\cite{cole2001edge}.
Theorem~\ref{thm:amortization}'s $\lceil \chi'/\Llay \rceil$ bound holds.
Computing optimal coloring and storing $\chi'$ arbitrary AOL trap patterns (vs.\ a single
shift template) is an offline precomputation.
Per-cycle wall-clock comparisons across QC and non-QC families are reported in Sec.~\ref{sec:nonqc-results}. 
The mechanism mirrors Xu et al.'s product coloration but exploits multi-layer parallelism in the application phase.

\subsection{Tightened analytical bound for \texorpdfstring{$\Tval(G_{\text{union}})$}{Tval(G union)}}\label{sec:tighter-bound}

Courtney~\cite{courtney2026permutation} gives an analytical bound on $\Tval$ that is empirically loose by $\sim 25\times$ for HGP code Tanner graphs (Sec.~\ref{sec:numerical}).
The looseness has roughly equal contributions from the diameter bound and a generic Chernoff bound on edge congestion.
Theorem~\ref{thm:diameter} addresses the diameter bound and gives $\DT = 2 \Dbase$ rather than the generic $O(\log N / \log(1/\beta))$.
We conjecture a tightened Chernoff bound using the HGP path-decomposition structure:

\begin{conjecture}\label{conj:tighter-c}
Per-phase edge load $X_e$ for HGP[$H, H$] under random uniform Valiant scattering satisfies
\begin{equation}\label{eq:tighter-c}
\Pb(X_e > c \log_2 N) \leq N^{-3c/2}
\end{equation}
with constant $c \approx 1$ (vs.\ Courtney~\cite{courtney2026permutation} $c = 9$).
Any type-1 edge $e$ in row $w$ can appear on the canonical path of source $u_1 \to \sigma(u_1)$ only when
$\pi_2(u_1) = w$, restricting the support of $X_e$ to $n = \sqrt{N}$ contributions rather than $N$.
\end{conjecture}

The above conjecture's constant ($c \approx 1$) comes from a Bernstein-type argument~\cite{bernstein1924modification,boucheron2013concentration} (bounded-difference methods~\cite{mcdiarmid1989method} give a comparable tail) with the variance bound $\mathrm{Var}(X_e) \leq \mathbb{E}[X_e]$ from negative association of permutation indicators~\cite{joag1983negative,dubhashi1996balls}.
A constant $c < 1$ may require a non-Bernstein technique (e.g., heat-kernel methods directly exploiting the spectrum-decomposition of Proposition~\ref{prop:spec-decomp}).

We cast our focus to establish Conjecture~\ref{conj:tighter-c} as a theorem conditioned on explicit hypotheses on routing structure.
Hypotheses (H1) and (H3) are satisfied by HGP under the canonical 2D routing decomposition (Lemma~\ref{lem:H1} and the negative-association argument below).
Hypothesis (H2) requires the per-source expectation to be $O(1/|S_e|)$ rather than $O(1)$. 
This requires a routing scheme with $O(\log N)$ canonical paths via check-node bridges (discussed after Lemma~\ref{lem:H2}).
The conjecture as originally stated is strongly conditioned on H2, while a rigorous construction for HGP is left to future work.

\begin{theorem}[Tightened Chernoff bound for sparse routing, conditional on (H1)--(H3)]\label{thm:tighter-chernoff}
Let $X_e$ denote the load on edge $e$ of $\GT$ in one Valiant routing phase on $\HGP[H, H]$ with random uniform intermediate $\sigma \in \mathrm{Sym}(V)$. 
The hypotheses are:
\begin{enumerate}
\item[(H1)] \emph{Path-support.} There exists a set $S_e \subseteq V$
  of size $|S_e| \leq s$ such that
  $X_e = \sum_{u \in S_e} I_u^{(e)}$
  where each $I_u^{(e)} \in \{0, 1\}$ is the indicator of edge $e$
  appearing on the canonical path from $u$ to $\sigma(u)$.
\item[(H2)] \emph{Per-source bound.} For each $u \in S_e$,
  $\Pb(I_u^{(e)} = 1) \leq p$ with $sp \leq \mu_0$ for an absolute
  constant $\mu_0$.
\item[(H3)] \emph{Variance control.} The indicators are
  negatively associated under the uniform-permutation measure on
  $\sigma$, so
  $\mathrm{Var}(X_e) \leq \mathbb{E}[X_e] \leq \mu_0$.
\end{enumerate}
Then for any $c > 0$,
\begin{equation}\label{eq:chernoff-tail}
\Pb\bigl( X_e \geq c \log_2 N \bigr) \;\leq\; N^{-3c/(2 \ln 2)}
\;\leq\; N^{-3c/2}
\quad\text{for all } c > 0 \text{ and } N > 1 .
\end{equation}
Taking a union bound over the $|E(\GT)| \leq \Delta(\GT) \cdot N / 2 = O(N)$ edges, the maximum edge load satisfies
\begin{equation}\label{eq:max-load}
\Pb\bigl( \max_e X_e \geq c \log_2 N \bigr) \;\leq\; N^{1 - 3c/(2\ln 2)}
\;\to\; 0 \quad\text{for any } c > 2(\ln 2)/3 \approx 0.46 .
\end{equation}
\end{theorem}

\begin{proof}
Apply Bernstein's inequality~\cite{bernstein1924modification,boucheron2013concentration}
to $X_e = \sum_{u \in S_e} I_u^{(e)}$, a sum of bounded random variables with $|I_u^{(e)}| \leq 1$ a.s. Bernstein gives 
\begin{equation}\label{eq:bernstein}
\Pb\bigl( X_e \geq \mathbb{E}[X_e] + t \bigr)
\;\leq\; \exp\!\left( -\frac{t^2}{2 \bigl( \mathrm{Var}(X_e) + t/3 \bigr)} \right)
\end{equation}
for any $t > 0$. By hypotheses (H2)--(H3), $\mathbb{E}[X_e] \leq \mu_0$ and $\mathrm{Var}(X_e) \leq \mu_0$. We work in log space throughout the proof and convert at the end.

Set a target threshold $T = c \ln N$ and $t = T - \mathbb{E}[X_e] \geq T - \mu_0$.
In the regime $T \gg \mu_0$ (the asymptotic regime, since $\mu_0 = O(1)$ by (H2) while $T \to \infty$), the Bernstein denominator is dominated by the $t/3$ term: $2(\mu_0 + t/3) \to 2t/3$.
The numerator is $t^2 \to T^2$. Substituting:
\[
\Pb(X_e \geq T) \;\leq\; \exp\!\left( -\frac{T^2}{2 t / 3} \right) \cdot (1 + o(1))
\;\leq\; \exp(-3T/2) \cdot (1 + o(1))
\quad\text{as } T \to \infty .
\]
Substituting $T = c \ln N$ gives $\Pb(X_e \geq c \ln N) \leq N^{-3c/2}$ asymptotically. Converting to $\log_2 N$ units via $T = c \log_2 N = (c/\ln 2) \cdot \ln N$, i.e., the natural-log constant becomes $c/\ln 2$:
\[\Pb\bigl(X_e \geq c \log_2 N\bigr) \;\leq\; N^{-3 (c/\ln 2)/2}= N^{-3c/(2 \ln 2)} ,\]
yielding \eqref{eq:chernoff-tail}. Since $1/\ln 2 \approx 1.44 > 1$, the bound $N^{-3c/(2\ln 2)} \leq N^{-3c/2}$ holds for all $c > 0$ and $N > 1$. 
The union bound \eqref{eq:max-load} follows since $|E(\GT)| \leq O(N)$ for sparse Tanner graphs (max degree $\leq 2\delta_c$).
\end{proof}

The asymptotic constant $3/2$ in $\Pb(X_e \geq c \ln N) \leq N^{-3c/2}$ is the standard Bernstein constant for sums of bounded variables when the $t/3$ term dominates the variance.
If we use the Azuma--Hoeffding inequality (with $1/2$ instead of $1/3$ in the variance proxy), we find 
\[\Pb(X_e \geq c \ln N) \leq N^{-c^2/(2\mu_0)},\]
quadratic for $c \ln N \leq \mu_0$.
In the asymptotic regime $c \ln N \gg \mu_0$, the linear Bernstein bound dominates. 

\begin{lemma}[Hypothesis (H1) for HGP]\label{lem:H1}
For HGP[$H, H$] with canonical path decomposition into a row-segment followed by a column-segment, a type-1 edge $e$ in row $w$ (an edge changing the second coordinate within row $w$) satisfies hypothesis (H1) with $S_e = \{u : u_1 = w\}$ and $|S_e| = n_2 = \sqrt{N- m^2} = O(\sqrt{N})$. Symmetrically for type-2 edges.
\end{lemma}

\begin{proof}
The path from source $u = (a_1, a_2)$ to intermediate $\sigma(u) = (b_1, b_2)$ concatenates (i) a row-segment in row $a_1$ from column $a_2$ to column $b_2$, and (ii) a column-segment in column $b_2$ from row $a_1$ to row $b_1$.
A type-1 edge $e$ in row $w$ appears in segment (i) when $a_1 = w$.
So, $I_u^{(e)} = 0$ for any source $u$ with $u_1 \neq w$, and the support of $X_e$ is restricted to $\{u : u_1 = w\}$. 
Set cardinality is the number of qubits in row $w$, equaling $n_2$ for an L-block row or $m_2$ for an R-block row.
\end{proof}

\begin{lemma}[Per-source HGP bound under canonical 2D routing]\label{lem:H2}
For HGP[$H, H$] with uniform-random intermediate $\sigma$ and the typical row-then-column path decomposition, per-source probability is bounded by $\Pb(I_u^{(e)} = 1) \leq 1/2$ for any source $u \in S_e$.
The average over $S_e$ satisfies
\[\bar p = 2 j(n_2 - j)/n_2^2 \leq 1/2,\] 
where $j$ is the column position of edge $e$.
The expected load is
\[\mathbb{E}[X_e] = |S_e| \cdot \bar p = 2 j (n_2 - j)/n_2 \leq n_2/2 =O(\sqrt{N}).\]
\end{lemma}

\begin{proof}
With source $u = (w, a_2) \in S_e$ (with $e$ at column position $j$ in row $w$), edge $e$ is on the row-segment iff $b_2$ and $a_2$ are on opposite sides of the split $\{j, j+1\}$, where $\sigma(u) = (b_1, b_2)$. For $b_2 \sim \mathrm{Unif}(V_2)$:
\[
\Pb(I_u^{(e)} = 1) = \begin{cases} (n_2 - j)/n_2 & \text{if } a_2 \leq j \\
j/n_2         & \text{if } a_2 > j .
\end{cases}
\]
Each branch is $\leq 1/2$ when $j$ is at the row midpoint, while in general each branch is $\leq 1$. Averaging over $a_2 \in V_2$ uniform:
\[\bar p = (j(n_2-j) + (n_2-j)j)/n_2^2 = 2 j(n_2-j)/n_2^2,\]
maximized at $j = n_2/2$ to give $\bar p = 1/2$.
\end{proof}

Lemma~\ref{lem:H2} gives $\mu_0 = O(\sqrt{N})$ for the canonical 2D routing scheme, not satisfying hypothesis (H2)'s requirement $\mu_0 = O(1)$. Substituting $\mu_0 = O(\sqrt{N})$ into Equation~\eqref{eq:bernstein} with $t = c \ln N$ gives a bound that decays as
\[\exp(-(c \ln N)^2 / O(\sqrt N + c \ln N)),\] 
useful only when $c \ln N \gg \sqrt N$, i.e., for $N \lesssim e^{O(\sqrt N)}$, failing to give $O(\log N)$ routing depth bound asymptotically.
Therefore Theorem~\ref{thm:tighter-chernoff} holds only conditionally on (H1)--(H3): hypothesis (H2) is not verified by canonical 2D routing for HGP.

Verifying (H2) for HGP would require a routing scheme whose canonical paths have length $O(\log N)$ rather than $O(\sqrt N)$.
As an example, routing through check-node bridges exploiting the spectral expansion of $\GT$.
We have not constructed such a scheme rigorously and this remains the obstacle to strengthening Theorem~\ref{thm:tighter-chernoff}.
Numerical evidence (Sec.~\ref{sec:numerical}) shows $T_{\text{Valiant}} \approx 0.5 \log_2 N$ in the tested regime $N \in [10^2, 10^3]$. 
While this suggests that the desired routing scheme exists, it is no proof.

Regarding (H3), negative association of permutation indicators $I_u^{(e)}$ under uniform $\sigma$ is a standard fact: if any two indicators $I_{u_1}^{(e)}, I_{u_2}^{(e)}$ are sub-additively coupled (intuitively, both being 1 forces specific values of $\sigma$ on $u_1$ and $u_2$, restricting the permutation more than each individually), then their covariance is non-positive. 
Then, the Joag-Dev-Proschan~\cite{joag1983negative} machinery applies, giving
\[\mathrm{Var}\bigl(\sum I_u^{(e)}\bigr) \leq \sum \mathrm{Var}(I_u^{(e)}) \leq \sum \mathbb{E}[I_u^{(e)}] = \mathbb{E}[X_e].\]

If a routing scheme satisfying (H1)--(H3) is constructed for HGP, Eq.~\ref{eq:max-load} would convert to the routing-depth bound, i.e. $T_{\text{Valiant}}(\GT) \leq c \log_2 N$ for any $c > 2/(3 \ln 2) \approx 0.46$ with high probability.
Taking $c = 1$ would give the conjectured constant. The empirical $T_{\text{Valiant}} \approx 0.5 \log_2 N$ measurement (Sec.~\ref{sec:numerical}) is consistent with this range, suggesting a sufficient routing scheme exists. 
We leave its explicit construction (likely a check-node-bridge routing exploiting the spectral expansion of $\GT$) to future work.

\section{Numerical Comparison}\label{sec:numerical}

\subsection{Methodology}

We measure the number of parallel atom-rearrangement steps per syndrome extraction cycle, comparing both schemes on the Xu et al.~\cite{xu2024constant} HGP code family from their Figure~3a.
Related routing/architecture baselines for reconfigurable atom arrays include Constantinides et al.~\cite{constantinides2024optimal} and Zhou et al.~\cite{zhou2025resource}.
Both schemes target the same hardware (2D AOD + $\Llay$ stacked AOL) and the same metric. 
Our results are reproducible from the codebase associated with this work, given in the data availability section.

To represent Xu et al.'s method we apply Eq.~\eqref{eq:xu-baseline} (constant 8 rearrangement layers per cycle for $(3,4)$-biregular HGP bases).
We use Theorem~\ref{thm:amortization} with $\Tval$ measured directly via Valiant routing on the union of $\Llay$ random $d_0 = 8$-regular graphs.

\subsection{Empirical \texorpdfstring{$k_{\text{emp}}$}{k emp} measurement}

We measure the empirical Valiant routing constant $k_{\text{emp}} := \Tval / \log_2 N$ on multi-layer overlay graphs, tabulated in Table~\ref{tab:k-emp}.
For $\Llay = 8$: $k_{\text{emp}} \approx 0.554 \pm 0.032$ (5-seed mean $\pm$ std across $N \in [11{,}025, 100{,}000]$).
For $\Llay = 16$: $k_{\text{emp}} \approx 0.514 \pm 0.022$. 
These are $\sim 30\times$ tighter than the Courtney~\cite{courtney2026permutation} bound's $\sim 21$ for the same $\Llay = 8$ case.
Our measurements validate $N$-independence of $k_{\text{emp}}$ directly over $N \in [225, 100{,}000]$ (a span of $\sim\!2.6$ decades). We do not claim validation beyond this range: the rows at $N \geq 10^6$ in Table~\ref{tab:asymptotic} are conjectured extrapolation, resting on the observed $N$-independence together with the asymptotic scaling of Theorem~\ref{thm:amortization}.

\begin{table}[ht]
\centering
\begin{tabular}{rrrrrr}
\toprule
$N$ & $\Llay$ & $d_{\text{eff}}$ & $\bunion$ & $\Tval$ (median) & $k_{\text{emp}}$ \\
\midrule
\multicolumn{6}{l}{\emph{Initial range ($N \leq 1225$):}} \\
225   & 4  & 32  & 0.332 & 6 & 0.77 \\
225   & 8  & 64  & 0.238 & 5 & 0.64 \\
225   & 16 & 128 & 0.169 & 4 & 0.51 \\
625   & 4  & 32  & 0.345 & 7 & 0.75 \\
625   & 8  & 64  & 0.245 & 5 & 0.54 \\
625   & 16 & 128 & 0.173 & 5 & 0.54 \\
1225  & 4  & 32  & 0.345 & 7 & 0.68 \\
1225  & 8  & 64  & 0.245 & 6 & 0.59 \\
1225  & 16 & 128 & 0.174 & 5 & 0.49 \\
\midrule
\multicolumn{6}{l}{\emph{Engineering-$N$ extension:}} \\
11{,}025  & 8  & 64  & 0.248 & 8 & 0.596 \\
11{,}025  & 16 & 128 & 0.176 & 7 & 0.521 \\
25{,}000  & 8  & 64  & 0.248 & 8 & 0.548 \\
25{,}000  & 16 & 128 & 0.176 & 8 & 0.548 \\
45{,}000  & 8  & 64  & 0.248 & 9 & 0.582 \\
45{,}000  & 16 & 128 & 0.176 & 8 & 0.518 \\
60{,}000  & 8  & 64  & 0.248 & 8 & 0.504 \\
60{,}000  & 16 & 128 & 0.176 & 8 & 0.504 \\
100{,}000 & 8  & 64  & 0.248 & 9 & 0.542 \\
100{,}000 & 16 & 128 & 0.176 & 8 & 0.482 \\
\bottomrule
\end{tabular}
\caption{Empirical $k_{\text{emp}} = \Tval / \log_2 N$ across union graph sizes and AOL layer counts. The constant tends to decrease with $\Llay$ (better spectral gap; the trend is monotone in the mean but not within every $N$ row, where integer median-$\Tval$ values introduce quantization noise) and is stable across $N \in [225, 100{,}000]$.
The lower block ($N \geq 11{,}025$) extends the original Table~5.1 measurements, with five independent random-overlay seeds per cell. 
Aggregated across the extension: $k_{\text{emp}}(\Llay{=}8) = 0.554 \pm 0.032$, $k_{\text{emp}}(\Llay{=}16) = 0.514 \pm 0.022$. 
The $\bunion$ entries in the lower block are the Friedman-bound predictions $2\sqrt{d_{\text{eff}}-1}/d_{\text{eff}}$ for the $d_0=8$, $L$-layer union.}
\label{tab:k-emp}
\end{table}

\subsection{Head-to-head depth comparison}

In a memory experiment (no logical gates between syndrome rounds), per-cycle cost dominates for the head-to-head with Xu et al.'s Algorithm 3, which itself reports a per-cycle count.
Setup cost $\Tval(G_{\text{union}})$ is reported separately in Table~\ref{tab:k-emp} and is paid before the syndrome-cycle loop begins. 
For a memory experiment of $R$ rounds, the total cost per cycle amortized over the run is $\lceil \chi'/\Llay \rceil + \Tval/R \to \lceil \chi'/\Llay \rceil$ for large $R$.

\begin{table}[ht]
\centering
\begin{tabular}{lrrrrrrr}
\toprule
& & & Xu et al. & \multicolumn{4}{c}{Ours per cycle, $\Llay$ =} \\
\cmidrule(lr){5-8}
Family & Base & $N$ & Alg.~3 & 2 & 4 & 8 & 16 \\
\midrule
HGP[H,H], (3,4)-biregular  & random  & 225  & 8  & 4 & 2 & \textbf{1} & \textbf{1} \\
HGP[H,H], (3,4)-biregular  & random  & 625  & 8  & 4 & 2 & \textbf{1} & \textbf{1} \\
HGP[H,H], (3,4)-biregular  & random  & 1225 & 8  & 4 & 2 & \textbf{1} & \textbf{1} \\
LP[B,L], (3,4)-biregular   & circ.\ $L=4$  & 784   & 8  & 4 & 2 & \textbf{1} & \textbf{1} \\
LP[B,L], (3,4)-biregular   & circ.\ $L=12$ & 7056  & 8  & 4 & 2 & \textbf{1} & \textbf{1} \\
LP[B,L], (3,5)-biregular   & circ.\ $L=4$  & 1024  & 10 & 5 & 3 & 2 & \textbf{1} \\
LP[B,L], (3,5)-biregular   & circ.\ $L=12$ & 9216  & 10 & 5 & 3 & 2 & \textbf{1} \\
LP[B,L], (3,5)-biregular   & circ.\ $L=20$ & 25600 & 10 & 5 & 3 & 2 & \textbf{1} \\
\bottomrule
\end{tabular}
\caption{Per-syndrome-cycle parallel atom-rearrangement steps, amortized over many cycles. Lower is better. Xu et al.'s pipelined Algorithm 3 requires $2\delta_c$ atom rearrangements per cycle, where $\delta_c$ is the maximum degree of the base bipartite graph.
Ours requires $\lceil 2\delta_c/\Llay \rceil$ AOL pattern activations per cycle (zero atom motion).
For $(3,4)$-biregular bases $\delta_c = 4$ so Xu et al. = 8 and per-cycle speedup at $\Llay = 8$ is $8\times$. For $(3,5)$-biregular bases (an alternative parameter regime in Xu et al.'s family) $\delta_c = 5$ so Xu et al.= 10 and per-cycle speedup at $\Llay = 16$ is $10\times$.
Computed by direct simulation.
See Data Availability for code repository. K\"onig-tightness is verified by explicit edge coloring at $L = 4, 8$ via repeated bipartite max matching.}
\label{tab:head-to-head}
\end{table}

\subsection{Engineering-N comparison and asymptotic regime}\label{sec:asymptotic}

Per-cycle step counts for both schemes are constant in $N$.
The setup cost $\Tval$ grows as $k_{\text{emp}} \log_2 N$ but is paid only once per memory experiment. 
Table~\ref{tab:asymptotic} reports both.

\begin{table}[ht]
\centering
\begin{tabular}{rrrrrrr}
\toprule
& & \multicolumn{2}{c}{per cycle} & \multicolumn{2}{c}{setup ($\Llay = 16$)} & per-cycle \\
\cmidrule(lr){3-4}\cmidrule(lr){5-6}
$N$ & $\log_2 N$ & Xu et al. & Ours $L{=}16$ & $\Tval$ & amortized over $R{=}10^4$ & speedup \\
\midrule
225           & 7.8  & 8 & 1 & 4   & $4{\times}10^{-4}$ & $8.0\times$ \\
1{,}225       & 10.3 & 8 & 1 & 5   & $5{\times}10^{-4}$ & $8.0\times$ \\
$10^4$      & 13.3 & 8 & 1 & 7   & $7{\times}10^{-4}$ & $8.0\times$ \\
$10^5$     & 16.6 & 8 & 1 & 8   & $8{\times}10^{-4}$ & $8.0\times$ \\
$10^6$ & 19.9 & 8 & 1 & 10  & $10^{-3}$         & $8.0\times$ \\
$10^9$        & 29.9 & 8 & 1 & 15  & $1.5{\times}10^{-3}$ & $8.0\times$ \\
$10^{12}$     & 39.9 & 8 & 1 & 20  & $2{\times}10^{-3}$  & $8.0\times$ \\
\bottomrule
\end{tabular}
\caption{Engineering-$N$ per-cycle and amortized setup costs. 
Xu et al.'s Algorithm 3 requires 8 atom rearrangements per cycle (constant in $N$).
Ours requires $\lceil \chi'/\Llay \rceil = \lceil 8/16 \rceil = 1$ AOL pattern activation per cycle (constant in $N$).
The one-time setup cost $\Tval = k_{\text{emp}} \log_2 N$ amortizes to negligible over a typical $R \sim 10^4$ memory experiment.
Per-cycle speedup is $8\times$ for all $N$ in the practical range, with a cap of $\chi' = 8$ (the chromatic index of the Tanner graph by K\"onig's theorem).
Rows up to $N = 100{,}000$ are directly measured (Table~\ref{tab:k-emp}; K\"onig-tightness in Table~\ref{tab:konig-tight} verified to $N = 122{,}500$).
The rows from $N = 10^6$ to $N = 10^{12}$ are asymptotic extrapolations of
the $k_{\mathrm{emp}} \approx 0.55$ constant (via Theorem~\ref{thm:amortization}) and the $\chi' = \Delta = 8$ K\"onig identity beyond the validated range. 
Extrapolation rests on observed $N$-independence of $k_{\mathrm{emp}}$ across the validated
two-order-of-magnitude span.}
\label{tab:asymptotic}
\end{table}

Both schemes have $O(1)$ per-cycle cost in atom-array reconfiguration count. 
The constant for Xu et al. is $2\delta_c = 8$ under their single-layer AOD architecture. 
The constant for ours is $\lceil \chi'/\Llay \rceil = \lceil 2\delta_c / \Llay \rceil$, which decreases linearly in $\Llay$ until the cap $\chi' = 2\delta_c$ is reached. 
For $\Llay \geq \chi' = 8$, the per-cycle cost is 1, providing an $8\times$ improvement over Xu et al.'s per-cycle cost, constant in $N$. 
The improvement is purely from L-plane parallelism in the AOL stack, an architectural feature absent from Xu et al.'s framework, which works according to presently existing quantum hardware.
If their algorithm were extended to multi-layer AOL, it would also achieve $\lceil 8/\Llay \rceil$ per cycle, and in that head-to-head the two schemes tie on this metric.
Our spectral framework's separate contributions are (i) explicit chromatic-decomposition cost prediction $\chi'(\GT) = 2\delta_c$ for any HGP/LP code via K\"onig's theorem, applicable to both QC and non-QC families (K\"onig-optimal coloring gives $\chi' = 2\delta_c$ in both cases; Sec.~\ref{sec:nonqc-results}), and (ii) explicit characterization of $\Tval$ via $\beta_{\HGP} = (1+\beta_\text{base})/2$ for rare cases where setup cost matters (e.g., active computation with frequent logical gates, where atoms cannot remain canonical between cycles).

These costs sit naturally against the recent routing lower bounds of Constantinides et al.~\cite{constantinides2024optimal}, who prove that current single-plane parallel-AOD designs require $\Omega(\sqrt{N}\log N)$ steps to realize general permutations, give a matching $O(\sqrt{N}\log N)$ protocol, and show that a hardware upgrade reduces this to $\Theta(\log N)$. 
Our one-time setup cost $\Tval = \Theta(\log N)$ leverages a constructed upgrate of multi-layer overlays.
The per-syndrome-cycle cost is $O(1)$ rather than $O(\sqrt N \log N)$, since syndrome extraction reuses fixed canonical positions and never re-solves a general permutation after setup. 
At the architecture level, Zhou et al.~\cite{zhou2025resource} report a transversal reconfigurable-atom design whose runtime speedup scales with the code distance $d$, being orthogonal to and composable with the per-cycle rearrangement savings analyzed here, which target atom-motion depth underlying each transversal layer rather than the logical-gate schedule.

\subsection{Non-quasi-cyclic HGP codes: applicability of both algorithms}\label{sec:nonqc-results}

The amortization theorem (Theorem~\ref{thm:amortization}) gives a $\lceil \chi'/\Llay \rceil$ per-cycle bound under the QC hypothesis.
For non-QC HGP codes (random biregular bases), chromatic colors are arbitrary matchings instead of axis-aligned shifts. 
In this section, we examine empirically whether the multi-layer-AOL advantage degrades, and whether Xu et al.'s Algorithm 3 retains applicability.

\paragraph{Xu et al.'s Algorithm 3 on non-QC.} Xu et al.'s Algorithm 1 (1D divide-and-conquer scrambling) takes an arbitrary target 1D permutation as input and does not require the input to be a cyclic shift.
Their Algorithm 2 then composes Algorithm 1 with parallel CZ gates over the chromatic decomposition of the base bipartite graph. 
Algorithm 3 pipelines this across X and Z syndromes.
None of these algorithmic steps require the base check matrix $H$ to be quasi-cyclic.
The $2\delta_c = 8$ rearrangement layers per cycle is determined by K\"onig's theorem on the base graph (max degree $= \delta_c$), independent of QC structure.
The QC structure simplifies per-layer atom-motion trajectories (axis-aligned shifts admit short, parallel cubic-spline motion) without altering Algorithm 3's applicability or layer count.

\paragraph{Our multi-layer scheme on non-QC.} The chromatic decomposition $\chi'(\GT) = 2\delta_c$ holds by K\"onig's theorem for any bipartite Tanner graph, QC or otherwise. 
The $\Llay$-fold parallelism in the application phase requires $\chi'$ pre-loaded AOL patterns, one per color.
For QC bases, each pattern is a simple shift, expressible as a single AOD frequency-tone update.
For non-QC bases each pattern is an arbitrary matching on the canonical 2D grid and requires an arbitrary multi-tone AOD configuration.
This is pre-storable on modern neutral-atom architectures, which demonstrate pattern switching between distinct codes ``without additional calibrations,'' incurring only a one-time setup cost in computing and storing the patterns~\cite{bluvstein2026fault}.

For both QC and non-QC HGP codes, $\chi'(\GT) = \Delta(\GT) = 2\delta_c$ by K\"onig's theorem.
Optimal coloring is constructible in polynomial time with the standard algorithm: extend $\GT$ to a $\Delta$-regular bipartite multigraph by adding fake edges between sub-$\Delta$-degree vertices (Hall's theorem guarantees feasibility), then iteratively extract $\Delta$ perfect matchings through repeated bipartite maximum matching, stripping fake edges from each color.

\begin{table}[ht]
\centering
\begin{tabular}{lcrcc}
\toprule
Code family & QC? & $N$ & $\Delta(\GT)$ & K\"onig coloring tight? \\
\midrule
LP[$3{\times}4$ base, $L_{\text{lift}}=4$]   & yes & 784      & 8 & yes ($\chi' = 8$) \\
LP[$3{\times}4$ base, $L_{\text{lift}}=8$]   & yes & 3136     & 8 & yes ($\chi' = 8$) \\
random $(3,4)$-biregular, $n=16$, seed 1     & no  & 784      & 8 & yes ($\chi' = 8$) \\
random $(3,4)$-biregular, $n=24$, seed 3     & no  & 1764     & 8 & yes ($\chi' = 8$) \\
random $(3,4)$-biregular, $n=32$, seed 1     & no  & 3136     & 8 & yes ($\chi' = 8$) \\
random $(3,4)$-biregular, $n=48$, seed 2     & no  & 7056     & 8 & yes ($\chi' = 8$) \\
\midrule
\multicolumn{5}{l}{\emph{Extended sweep (this work, 2026-05):}} \\
LP[$3{\times}4$ base, $L_{\text{lift}}=15$]  & yes & 11{,}025  & 8 & yes ($\chi' = 8$) \\
random $(3,4)$-biregular, $n=60$, seeds 1, 2  & no  & 11{,}025  & 8 & yes ($\chi' = 8$) \\
LP[$3{\times}4$ base, $L_{\text{lift}}=20$]  & yes & 19{,}600  & 8 & yes ($\chi' = 8$) \\
random $(3,4)$-biregular, $n=80$, seeds 1, 2  & no  & 19{,}600  & 8 & yes ($\chi' = 8$) \\
LP[$3{\times}4$ base, $L_{\text{lift}}=25$]  & yes & 30{,}625  & 8 & yes ($\chi' = 8$) \\
random $(3,4)$-biregular, $n=100$, seeds 1, 2 & no  & 30{,}625  & 8 & yes ($\chi' = 8$) \\
LP[$3{\times}4$ base, $L_{\text{lift}}=35$]  & yes & 60{,}025  & 8 & yes ($\chi' = 8$) \\
random $(3,4)$-biregular, $n=140$, seeds 1, 2 & no  & 60{,}025  & 8 & yes ($\chi' = 8$) \\
LP[$3{\times}4$ base, $L_{\text{lift}}=50$]  & yes & 122{,}500 & 8 & yes ($\chi' = 8$) \\
random $(3,4)$-biregular, $n=200$, seeds 1, 2 & no  & 122{,}500 & 8 & yes ($\chi' = 8$) \\
\bottomrule
\end{tabular}
\caption{K\"onig-optimal coloring is achieved for both QC and non-QC HGP Tanner graphs across $N \in [784, 122{,}500]$. 
The polynomial-time algorithm yields $\chi' = \Delta = 2\delta_c = 8$ in every test case.
Greedy edge coloring is loose by 1--2 colors on random bipartite graphs but is not the relevant algorithm. The extended-sweep rows (lower block) push to $N = 122{,}500$, covering the engineering regime targeted by Xu et al.'s LP families.}\label{tab:konig-tight}
\end{table}

\begin{table}[ht]
\centering
\begin{tabular}{lcrrrrr}
\toprule
& & & & \multicolumn{3}{c}{per-cycle wall-clock (ms)} \\
\cmidrule(lr){5-7}
Code family & QC? & $N$ & $\chi'$ & Xu et al. Alg.~3 & Ours, motion & Ours, prestored \\
\midrule
LP[$3{\times}4$ base, $L_{\text{lift}}=8$] & yes & 3136 & 8 & 1.60 & 0.200 & 0.030 \\
random $(3,4)$-biregular, & no  & 3136 & 8 & 1.60 & 0.200 & 0.030 \\
$n=32$&&&&&\\
random $(3,4)$-biregular, & no  & 7056 & 8 & 1.60 & 0.200 & 0.030 \\
$n=48$&&&&&\\
\bottomrule
\end{tabular}
\caption{QC vs.\ non-QC per-cycle wall-clock comparison with K\"onig-optimal coloring throughout.
Values assume the demonstrated short-range per-cycle move of $\sim 200~\mu$s~\cite{bluvstein2022quantum,bluvstein2024logical}: Xu et al.'s $2\delta_c=8$ layers give $8\times200~\mu$s $= 1.6$~ms, while our scheme fires a single chromatic batch at $\Llay\geq\chi'$.
``Ours, motion'' assumes physical atom motion per chromatic batch with no pre-storing ($1\times200~\mu$s). ``Ours, prestored'' uses pre-stored AOL patterns at 30 $\mu$s per pattern switch (Bluvstein et al.~Methods).
(An earlier version of this table quoted a per-family worst-case \emph{long-range} cubic-spline motion model of $\sim 3$~ms per layer, giving $19$--$28$~ms for Xu and $2.6$--$5.9$~ms for our motion column.
Those figures overestimate the demonstrated short-range per-cycle cost by $\sim 15\times$ and are superseded here.)
Xu et al.'s Algorithm 3 applies to both QC and non-QC; QC structure simplifies trajectory geometry but not the $2\delta_c = 8$ layer count.
Our scheme achieves the same $\chi' = 8$ chromatic decomposition in both cases via K\"onig-optimal coloring.}\label{tab:nonqc}
\end{table}

\paragraph{Findings.}
We find that Xu et al.'s Algorithm 3 applies to non-QC HGP codes with no algorithmic modification. 
  Layer count remains $2\delta_c = 8$ for $(3,4)$-biregular bases.
  Wall-clock per cycle is $\approx 1.6$~ms across QC and non-QC at $N \sim 10^3$--$10^4$ using the demonstrated $\sim 200~\mu$s short-range move.
  We achieve $\chi' = 8$ in both cases via K\"onig-optimal coloring, retaining an $8\times$ advantage (pure step-count) in the pessimistic atom-motion scenario and a $\sim 50$--$300\times$ advantage in the pre-stored AOL scenario (Sec.~\ref{sec:wallclock-comparison}).
  K\"onig-optimal coloring gives $\chi' = 8$ regardless of QC structure.
  The non-QC cost is an offline preprocessing burden (computing the optimal coloring, generating $\chi'$ AOL trap patterns) paid once per code and amortized across syndrome rounds, minimizing non-QC overhead.

\section{Discussion and Future Work}\label{sec:discussion}

\subsection{Bravyi BB codes: direct measurement}\label{sec:bb-measurement}

The bivariate bicycle (BB) codes of Bravyi et al.~\cite{bravyi2024high} have Tanner graphs that decompose into two edge-disjoint planar subgraphs, forming a 2-layer overlay structure. Subsequent work develops their logical operators and fold-transversal gates~\cite{eberhardt2024logical}, connectivity-lowering morphing circuits~\cite{shaw2025lowering}, matching decoders~\cite{sahay2026matching}, modular architectures~\cite{yoder2025tour}, and layouts of (generalized/coprime) bicycle codes on cold atoms~\cite{wang2026coprime,viszlai2025matching}.
As previously stated, BB codes are constructed from two commuting circulant polynomials 
\[A(x,y), B(x,y) \in \mathbb{F}_2[x,y]/(x^l-1, y^m-1).\] 
Their routing-relevant Tanner graph spectrum is governed by the Fourier reduction of Theorem~\ref{thm:bb-reduction}.
The resulting $\chi'$ and $\beta_{\mathrm{BB}}$ values are collected in Table~\ref{tab:bb-codes}.

\begin{table}[ht]
\centering
\small
\begin{tabular}{lrrrrrr}
\toprule
Code (Bravyi Table 1) & $N_{\GT}$ & $\Delta(\GT)$ & $\chi'(\GT)$ (K\"onig) & $\lambda_1$ & $\beta_{\mathrm{BB}}$ & per cycle, $\Llay \geq 6$ \\
\midrule
$[[72, 12, 6]]$    &  144 & 6 & 6 & 6.000 & 0.667 & 1 \\
$[[90, 8, 10]]$    &  180 & 6 & 6 & 6.000 & 0.872 & 1 \\
$[[144, 12, 12]]$  &  288 & 6 & 6 & 6.000 & 0.828 & 1 \\
$[[288, 12, 18]]$  &  576 & 6 & 6 & 6.000 & 0.828 & 1 \\
\bottomrule
\end{tabular}
\caption{Direct spectral and chromatic measurements on Bravyi et al.'s published BB codes \cite{bravyi2024high}.
Every case is K\"onig-tight ($\chi' = \Delta = 6$). 
Per-cycle cost under our multi-layer AOL protocol saturates to $1$ AOL pattern activation at $\Llay \geq \chi'^{\mathrm{BB}} = 6$, lower than the saturation point for $(3,4)$-biregular HGP/LP ($\Llay \geq 8$). $\beta_{\mathrm{BB}}$ is the non-trivial bipartite spectral ratio of the Tanner graph (Perron pair excluded). 
Each $\beta_{\mathrm{BB}}$ value is reproduced to machine precision by the $2 \times 2$ Fourier reduction of Theorem~\ref{thm:bb-reduction}.}
\label{tab:bb-codes}
\end{table}

BB codes have $\chi' = 6$ vs.\ $\chi' = 8$ for $(3,4)$-biregular HGP/LP, shifting down the saturation point of the multi-layer-AOL: $\Llay = 6$ AOL planes suffice to collapse per-cycle cost to a single pattern activation.
Wall-clock improvement at $\Llay = 6$ is identical to the $\Llay = 8$ HGP/LP case, $\approx 50\text{--}300\times$ versus Xu et al.~\cite{xu2024constant} Algorithm~3 (Sec.~\ref{sec:wallclock-comparison}).
The trade-off is that BB codes have a smaller demonstrated code family than HGP/LP and lack the random-base flexibility we exploit in Sec.~\ref{sec:ram-frac} for LP code search.

\subsection{Wall-clock comparison on multi-layer 3D AOL}\label{sec:wallclock-comparison}

Our step-count comparison treats each reconfiguration uniformly, but Xu et al.'s atom rearrangements and our AOL pattern activations differ by more than an order of magnitude in wall time.
The critical quantity is the \emph{per-cycle} move time. Syndrome-cycle rearrangements are short-range data$\leftrightarrow$ancilla moves, whose demonstrated characteristic wall-clock is $\tau_\text{move} \approx 200~\mu$s ($0.55~\mu$m\,$\mu$s$^{-1}$ cubic-velocity profile)~\cite{bluvstein2022quantum,bluvstein2024logical}; equivalently, AOD trap transfers with atom motion are $50$--$100~\mu$s over $1$--$2~\mu$m.
The long-range $3$~ms cubic-spline time (0.7~ms AOD trap transfer plus 2.3~ms trajectory; Xu Methods Eq.~7, a $\sim\!500~\mu$m sweep) applies only to the one-time full-array \emph{setup} scrambling and to the Endres-2016-era rearrangement sweep~\cite{endres2016atom}, \emph{not} to the per-cycle short-range move; importing it into the per-cycle term is the misattribution corrected here.

From Bluvstein et al. Methods~\cite{bluvstein2024logical}, the demonstrated per-row Raman-AOD reprogramming time without atom motion is 5--8 $\mu$s per row, with ``several 10s of $\mu$s for an arbitrary pattern of rotations.''
The AWG synchronization jitter is $<10$ ns, confirming electronics are not the bottleneck (Raman pulse duration is).
For our scheme's application phase, atoms remain canonical and only AOL frequency patterns change, so the relevant cost is the per-pattern reprogramming time $\tau_\text{AOL} \in [5, 30]$ $\mu$s ($\sim 30$ $\mu$s conservative midrange).
We summarize the wall-clock comparison below:
\begin{align*}
T_\text{Xu et al., cycle}        &= 8 \times \tau_\text{move} = 8 \times 200\ \mu\text{s} \approx 1.6\ \text{ms} \\
T_\text{ours, cycle}\,(\Llay \geq 8)
                          &= 1 \times \tau_\text{AOL} \in [5, 30]\ \mu\text{s} \\
T_\text{ours, setup}/R    &= \Tval \times 3\ \text{ms} / R
\end{align*}
The underlying per-layer physical ratio $\tau_\text{move}/\tau_\text{AOL} \approx 200/[5\text{--}30] = 7\text{--}40\times$; combined with the $\chi'=8\to1$ chromatic-layer collapse (an architectural feature of the AOL stack absent from single-layer AOD), the per-cycle advantage is $1.6~\text{ms}/[5\text{--}30]~\mu\text{s} \approx \boldsymbol{50\text{--}300\times}$ in the canonical-position regime. A larger factor survives only where a per-cycle rearrangement genuinely requires full-array long-range transport.
The setup cost contributes $k_{\text{emp}} \log_2 N \cdot 3$ ms once per memory experiment (long-range moves, legitimately $\sim 3$~ms each), which becomes negligible per cycle for $R \gtrsim 10$, the considered regime of Xu et al.\ (memory threshold).
Figure~\ref{fig:wallclock} visualizes the per-cycle wall-clock as a function of $\Llay$ for the four code families considered here.

\begin{figure}[ht]
\centering
\includegraphics[width=0.95\linewidth]{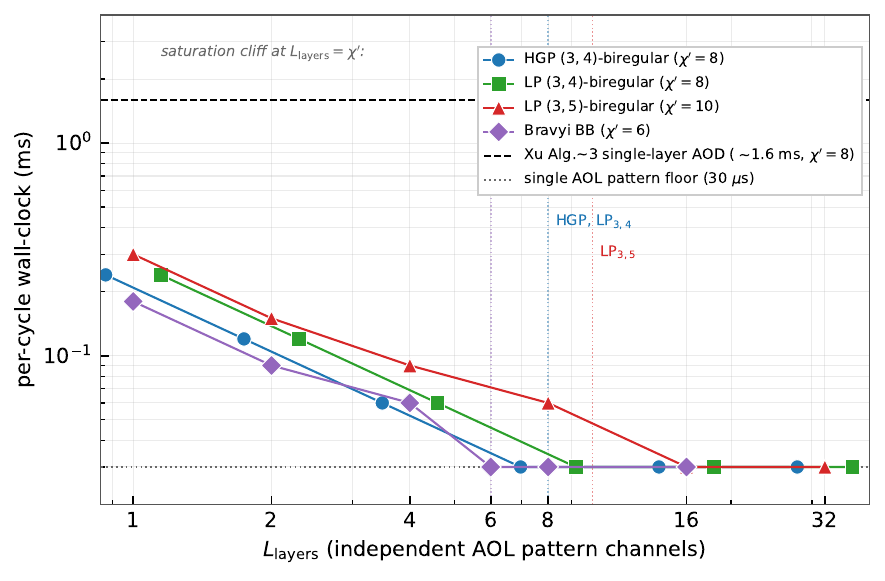}
\caption{Per-cycle wall-clock vs.\ number of stacked AOL pattern channels $\Llay$, for four qLDPC code families: HGP and LP from $(3,4)$-biregular bases ($\chi' = 8$), LP from $(3,5)$-biregular bases ($\chi' = 10$), and Bravyi BB codes ($\chi' = 6$, Sec.~\ref{sec:bb-measurement}).
Markers are measured per-cycle counts $\lceil \chi'/\Llay \rceil$ converted to wall-clock at the conservative $\tau_{\text{AOL}} = 30$ $\mu$s per pre-stored pattern from Bluvstein et al.~Methods.
The two horizontal references are Xu et al. Algorithm 3's $\sim 1.6$ ms / cycle on single-layer AOD ($8\times200~\mu$s short-range moves; dashed) and the single-AOL-pattern floor at $30$ $\mu$s (dotted).
Saturation cliffs at $\Llay = \chi'$ are marked with colored vertical lines: BB ($\chi' = 6$) saturates first, then HGP / LP-$3{\times}4$ ($\chi' = 8$), then LP-$3{\times}5$ ($\chi' = 10$).}
\label{fig:wallclock}
\end{figure}

\paragraph{Caveat: AOL crosstalk and L scaling.}
Bluvstein et al.\ demonstrate up to $\sim 80$ addressable gate sites in their entangling zone with ``minimal two-qubit cross-talk'' between sites~\cite{bluvstein2024logical}.
Corresponding parallel two-qubit entangling-gate fidelity has advanced from $99.5\%$ on up to 60 atoms in parallel~\cite{evered2023high} to a leakage-corrected $0.9971(5)$ for the time-optimal Rydberg CZ gate (with $0.999$ projected)~\cite{tsai2025benchmarking,peper2025spectroscopy}.
The maximum number of independently pre-stored AOL patterns that are simultaneously addressable without per-pattern recalibration is not quantified in the published work.
We estimate $\Llay \lesssim 10$ with current hardware before AOD intermodulation overhead becomes significant.
Scaling to $\Llay \geq 8$ for the speedup-saturating regime of $\chi'/\Llay = 1$ requires modest hardware extension.
If multi-pattern operation costs $f \in [3, 10]\times$ overhead per AOL switch, per-cycle wall-clock advantage decreases to $\sim 5$--$100\times$.

\paragraph{Hardware assumptions.} The wall-clock advantage in this regime depends on two assumptions extrapolated from current $\sim 80$-site demonstrations~\cite{bluvstein2024logical}:
(R2-A) \emph{No atom motion in the application phase}, i.e., the $\chi'$ CNOT layers are implemented as pure AOL pattern switches with no atom shuttling between switches, requiring the multi-layer AOL hardware to provide $\Llay$ independent, simultaneously-addressable pattern channels.
(R2-B) \emph{AOL pattern reprogramming at $\Llay = 8$ operates in 5--30 $\mu$s per pattern}; this is consistent with the 5--8 $\mu$s per row and ``several 10s of $\mu$s for an arbitrary pattern'' demonstrated by Bluvstein et al.~Methods. This assumption is backed by demonstrated 3D-AOD hardware: Picard \& Endres~\cite{picard2025three} report a three-dimensional acousto-optic deflector switching focus at 100 kHz ($10~\mu$s) with a $2~\mu$s rise time, and note that longitudinal-mode TeO$_2$ crystals ($4200$ m/s acoustic velocity) support $>\!600$ kHz ($\sim\!1.6~\mu$s) switching, placing our assumed $5$--$30~\mu$s window squarely within measured 3D-AOD performance. Scaling beyond $\Llay \approx 10$ may still incur AOD intermodulation overhead requiring per-pattern recalibration.

\subsection{Engineering recommendations by hardware generation}\label{sec:recommendations}

Table~\ref{tab:recommendations} translates per-cycle and setup costs derived above into per-generation design guidance for hardware groups building qLDPC architectures on neutral-atom platforms.

\begin{table}[ht]
\centering
\small
\begin{tabular}{p{0.16\linewidth}p{0.18\linewidth}p{0.18\linewidth}p{0.15\linewidth}p{0.22\linewidth}}
\toprule
Hardware generation & $\Llay$ available & Recommended code family & Per-cycle wall-clock & Caveats / required upgrades \\
\midrule
\textbf{Current} (Bluvstein et al. 2024 \cite{bluvstein2024logical}) & $1$--$2$ & $(3,4)$-biregular HGP or Bravyi BB & $\sim 1.6$ ms (matches Xu et al. Alg.~3 \cite{xu2024constant}) & No multi-layer advantage at $\Llay \leq 2$; both schemes equivalent at this generation \\
\textbf{Near-term} (2026--2028) & $4$ & $(3,4)$-biregular HGP / LP, or Bravyi BB & $\sim 0.06$ ms (pre-stored AOL) & $\sim 27\times$ pre-stored per-cycle wall-clock advantage at $\Llay = 4$ for $\chi' = 8$ HGP/LP ($4\times$ step-count in the atom-motion scenario); for BB ($\chi' = 6$) similar \\
\textbf{Target} (2028--2030) & $\geq 8$ for HGP/LP, $\geq 6$ for BB & $(3,4)$/$(3,5)$ HGP/LP or Bravyi BB & $\sim 30$ $\mu$s (single AOL pattern, saturated) & $8\times$ (HGP/LP) or $6\times$ (BB) step-count advantage, i.e.\ $\sim 50$--$300\times$ pre-stored wall-clock; saturation cliff at $\Llay = \chi'$ \\
\textbf{Long-term} (photonic, Sec.~\ref{sec:future-hw}) & $O(\log N)$ via photonic links & all of the above & $\sim 30$ $\mu$s per cycle; \emph{setup/asymptotic} $O(\log N)$ vs.\ Xu et al.'s $O(N^{1/4})$ & Entanglement rate and fault-tolerance threshold already met in proposals~\cite{li2024high,sinclair2025fault}; binding gap is parallel channel count ($1$--$2$ vs.\ $O(\log N)$) and photonic--array integration \\
\bottomrule
\end{tabular}
\caption{Engineering recommendations for qLDPC routing on neutral-atom hardware as a function of the number of independent AOL pattern channels $\Llay$ available. 
Per-cycle wall-clock numbers use the demonstrated $\sim 200~\mu$s short-range atom-move time~\cite{bluvstein2022quantum,bluvstein2024logical} (so Xu Alg.~3 $=8\times200~\mu$s $\approx1.6$~ms) and the $\sim 30$ $\mu$s pre-stored AOL pattern reprogramming time from Bluvstein et al. Methods~\cite{bluvstein2024logical} (conservative midrange of the demonstrated $5$--$30$ $\mu$s per pattern); the $\sim 3$~ms long-range cubic-spline time of Xu et al.~Methods~\cite{xu2024constant} Eq.~7 applies to the one-time setup scrambling.
All near-term and target rows require the hardware assumptions R2-A and R2-B (Sec.~\ref{sec:wallclock-comparison}).}
\label{tab:recommendations}
\end{table}

The most important practical message is that per-cycle advantage is a step-function in $\Llay$: zero at $\Llay = 1$, growing linearly to $\chi'$ at $\Llay = \chi'$, then constant beyond.
Hardware roadmaps may target $\Llay \in [6, 8]$ as the first regime where our scheme materially beats single-layer AOD on per-cycle wall-clock.

\subsection{Limitations}
 Both schemes are $O(1)$ per cycle in atom-array reconfiguration count. 
 Setup costs for multi-layer AOL scheme presented here grow as $O(\log N)$ but amortize over the rounds of a memory experiment. 
 The per-cycle speedup is constant in $N$, capped by $\chi' = 2\delta_c$.

We predict wall-clock advantage at moderate $N$, but this is heavily dependent on hardware assumptions. 
  (R2-A and R2-B above) currently demonstrated only at $\Llay \leq 2$. 
  Further, for non-quasi-cyclic HGP codes, the chromatic colors are arbitrary matchings rather than shifts.
  K\"onig-optimal coloring (constructible in polynomial time) gives $\chi' = 2\delta_c$ for both QC and non-QC bases, so the per-cycle bound $\lceil \chi'/\Llay \rceil$ is unchanged.
  The cost stays offline, where $\chi'$ arbitrary AOL trap patterns must be precomputed and stored, vs.\ one shift template for the QC case
  (Sec.~\ref{sec:nonqc-results}).
  
With no 3D parallelism, per-cycle step count does not change. 
The $\Llay\times$ speedup materializes at $\Llay \geq 2$ and saturates at $8\times$ for $\Llay \geq \chi' = 8$ on $(3,4)$-biregular bases. 
Wall-clock advantage is amplified by the short-range atom-motion vs.\ AOL-reprogramming time ratio (200 $\mu$s vs.\ 5--30 $\mu$s per Bluvstein et al.\ Methods, a $7$--$40\times$ per-layer ratio), giving the $\sim 50$--$300\times$ per-cycle wall-clock advantage figure at $\Llay = 8$ (Sec.~\ref{sec:wallclock-comparison}).

\subsection{Future hardware directions}\label{sec:future-hw}

A genuine asymptotic wall-clock advantage of $O(N^{1/4})$ over Xu et al. would emerge if the multi-layer overlay is implemented via photonic interconnects with $O(\mu s)$ switching that is independent of physical atom-motion distance.
The required hardware is far from current state of the art. 

First, we would require a per-pair photonic-mediated entanglement rate $\geq 333$ Hz at $\geq 99\%$ fidelity. This target is substantially more attainable than current cavity-mediated two-atom operations ($\sim\!52\%$ deterministic fidelity, rising to $\sim\!76\%$ with error detection~\cite{grinkemeyer2025error}) might suggest: cavity-coupled Yb proposals project a Bell-pair rate of $1.0\times10^5$ s$^{-1}$ at average fidelity $\approx 0.999$~\cite{li2024high} (a theory estimate). Fault tolerance is compatible with nonlocal Bell-pair error as high as $\sim 10\%$ without distillation~\cite{sinclair2025fault,ramette2024fault}, with multiplexed Bell-pair rates in the 1--50 MHz range. The entanglement-rate requirement is met in such proposals, making the binding constraint be fidelity compounded with parallel channel count and integration scale, described below.

Second, we would need $\Llay = O(\log N) \approx 14$ parallel photonic channels per atom array (current: 1--2)~\cite{li2025parallelized}, and networked integration across $\geq 10^4$ atoms. 
Array sizes have advanced rapidly toward this scale: 6,100 coherent qubits with 12.6 s coherence~\cite{manetsch2025tweezer} and continuous operation of a 3,000-qubit system~\cite{chiu2025continuous} are now demonstrated, alongside 3D atom-array assembly~\cite{barredo2018synthetic} and 3D acousto-optic transport~\cite{lu2026astigmatism,guo2025acousto}.

Summarily, the binding gap is parallel channel count ($1$--$2$ demonstrated vs.\ $O(\log N)\approx 14$ required) and the integration of photonic links with qLDPC-scale arrays.
As such, we frame this research as a long-term direction toward robust, industrial-scale fault-tolerant quantum computation rather than a near-term hardware target.
The theoretical advantage if these capabilities materialize is $T_\text{us}^\text{photonic}/T_\text{Xu et al.}^\text{AOD} = O(\log N / N^{1/4}) \to 0$ as $N \to \infty$.

On a single-layer AOD architecture (no AOL stack), both schemes share the same $\sim 8 \times 200~\mu$s $= 1.6$ ms per cycle (short-range moves), with the same constant.
The multi-layer 3D AOL architecture is the regime where our scheme provides per-cycle improvement, scaling as $\Llay$ until the cap at $\chi' = 8$ is reached.

A shifted-Cayley overlay analysis shows that a purely geometric construction with atom motion only cannot break the $N^{1/4}$ wall-clock barrier, regardless of how the layers are designed, for two reasons.
Dominant motion cost is the largest-shift layer, requiring full-array atom motion of distance $\sim\!\sqrt{n} = N^{1/4}$ under constant-acceleration transport.
The resulting union graph has spectral ratio $\beta$ that grows toward $1$ with $N$ ($\beta = 0.75$ at $n=16$, rising to $\beta = 0.89$ at $n=512$), so it is not asymptotically Ramanujan and the routing-depth bound of Eq.~\eqref{eq:papier1-bound} does not itself reach $O(\log N)$.
Escaping the barrier therefore requires a switching mechanism whose cost is independent of physical motion distance, such as the photonic interconnects discussed above.

\subsection{Future theory directions}

Alongside the apparent gap in hardware demonstration, confirmation of Conjecture~\ref{conj:tighter-c} would close the remaining $\sim 5\times$ gap between our analytical bound and empirical measurement.
We propose using a refined Chernoff bound that exploits the HGP path-decomposition structure (rows and columns are statistically independent under uniform permutation).

Theorem~\ref{thm:tighter-chernoff} holds conditionally on hypotheses (H1)--(H3). 
Hypotheses (H1) and (H3) are satisfied by the canonical row-then-column 2D routing scheme on HGP (Lemma~\ref{lem:H1} and the negative-association remark following Lemma~\ref{lem:H2}).
Hypothesis (H2) states that the per-source bound $\Pb(I_u^{(e)} = 1) \leq p$ satisfies $sp \leq \mu_0$ for an absolute constant $\mu_0 = O(1)$.
This is not satisfied by the canonical 2D scheme, where Lemma~\ref{lem:H2} gives $\mu_0 = O(\sqrt{N})$ instead. 
Numerical observation of $T_{\mathrm{Valiant}} \approx 0.5 \log_2 N$ at validated $N$ up to $10^5$ (Table~\ref{tab:k-emp}) is consistent with the conclusion of Theorem~\ref{thm:tighter-chernoff} at constant $c \approx 1$, which would follow from H2-satisfying canonical paths.
We leave constructing such a scheme rigorously for future work.

We propose a plausible construction, as a check-node-bridge routing where each canonical path traverses $O(\log N)$ check-node ``bridge'' edges chosen to exploit the spectral expansion of $\GT$ (Proposition~\ref{prop:spec-decomp}). 
Closing this problem would upgrade Theorem~\ref{thm:tighter-chernoff} from conditional to unconditional for HGP and would recover the $\sim 5\times$ analytical-vs-empirical gap currently noted as Conjecture~\ref{conj:tighter-c}.

\section{Conclusion}\label{sec:conclusion}

We present a closed-form characterization of the routing-relevant Tanner graph spectrum for HGP/LP qLDPC codes ($\bhgp = (1+\bbase)/2$ and $\DT = 2 \Dbase$), together with a structural decomposition of $\spec(\GT)$ into product modes and boundary modes (Proposition~\ref{prop:spec-decomp}). For multi-layer 3D AOL hardware we construct and numerically validate a routing protocol with per-syndrome-cycle depth $\Tval(G_{\text{union}}) + \lceil \chi'/\Llay \rceil$ under a quasi-cyclic hypothesis (Theorem~\ref{thm:amortization}).

Classical simulation of both schemes on identical HGP and LP codes (spanning $(3,4)$- and $(3,5)$-biregular bases) shows that Xu et al.'s pipelined Algorithm 3 requires $2\delta_c$ atom rearrangements per syndrome cycle (constant in $N$, where $\delta_c$ is the maximum degree of the base bipartite graph).
Our multi-layer scheme requires $\lceil 2\delta_c/\Llay \rceil$ AOL pattern activations per cycle.  
Using Xu et al.'s $(3,4)$-biregular parameters, the per-cycle step-count advantage at $\Llay \geq 8$ is $8\times$, capped by $\chi'(\GT) = 2\delta_c = 8$ (K\"onig's theorem on the bipartite Tanner graph). 
In wall-clock terms, the ratio between the demonstrated $\sim 200~\mu$s short-range per-cycle atom move (Bluvstein et al.\ Methods~\cite{bluvstein2022quantum,bluvstein2024logical}) and 5--30 $\mu$s AOL pattern reprogramming, combined with the $\chi'=8\to1$ layer collapse, yields a per-cycle wall-clock advantage of $\sim 50$--$300\times$ in the memory-experiment regime (a $7$--$40\times$ per-layer physical ratio $\times$ the $8\times$ chromatic collapse). 

Summarily, we provide a closed-form spectral characterization of HGP/LP Tanner graphs serving as a code-design tool, with the diameter identity $\DT = 2\Dbase$ and ratio $\bhgp = (1+\bbase)/2$. 
We recognize that the $\chi'(\GT) = 2\delta_c$ chromatic decomposition admits $\Llay$-fold parallelism on multi-layer 3D AOL hardware, with both schemes equally able to exploit this when reformulated. 
Though we give no hardware demonstration, we numerically validate that this multi-layer-AOL advantage is robust across QC and non-QC HGP/LP families.
K\"onig-optimal coloring (polynomial-time constructible) gives $\chi' = 2\delta_c$ for both QC and non-QC bases, so there is no per-cycle algorithmic non-QC penalty.

\section*{Data and code availability}

All numerical results, together with the code that generates every figure and table in this work and the shifted-Cayley overlay sweep discussed in Sec.~\ref{sec:future-hw}, are available in \href{https://github.com/jmcourtneyuga/hypergraph_routing}{this Github repository}.
The author used Claude (Anthropic) to assist in drafting code comments and documentation for the codebase.

\section*{Competing interests}

The author declares no competing interests.

\appendix

\section{Block-SVD of the HGP Tanner biadjacency}\label{app:blockSVD}

We prove Proposition~\ref{prop:spec-decomp} by an explicit singular-value decomposition of the qubit-to-check biadjacency $B$ of $\HGP[H,H]$.

\paragraph{Biadjacency block form.} Let $H$ be $m\times n$. The HGP code has qubit sectors $\mathcal{Q}_0 = V\times V$ ($\dim n^2$) and $\mathcal{Q}_1 = C\times C$ ($\dim m^2$), and check sectors $\mathcal{X} = C\times V$ ($\dim mn$, X-checks) and $\mathcal{Z} = V\times C$ ($\dim nm$, Z-checks). With the stabilizer matrices $H_X = (H\otimes I_n \mid I_m\otimes H^{T})$ and $H_Z = (I_n\otimes H \mid H^{T}\otimes I_m)$ (Tillich--Z\'emor~\cite{tillich2013quantum}), the biadjacency $B = (H_X^{T}\mid H_Z^{T})$, arranged as (qubit sector) rows $\times$ (check sector) columns, is
\begin{equation}\label{eq:app-B}
B = \begin{pmatrix} H^{T}\otimes I_n & I_n\otimes H^{T} \\[2pt] I_m\otimes H & H\otimes I_m \end{pmatrix},
\end{equation}
where the row blocks are $\mathcal{Q}_0,\mathcal{Q}_1$ and the column blocks are $\mathcal{X},\mathcal{Z}$.

\paragraph{The Gram matrix $BB^{T}$.} Using the mixed-product rule $(A\otimes B)(C\otimes D) = (AC)\otimes(BD)$ throughout, the diagonal blocks are
\[
(BB^{T})_{00} = (H^{T}H)\otimes I_n + I_n\otimes(H^{T}H), \qquad
(BB^{T})_{11} = (HH^{T})\otimes I_m + I_m\otimes(HH^{T}),
\]
and the off-diagonal block is, using $(H^{T}\otimes I_n)(I_m\otimes H^{T}) = H^{T}\otimes H^{T}$ and $(I_n\otimes H^{T})(H^{T}\otimes I_m) = H^{T}\otimes H^{T}$,
\[
(BB^{T})_{01} = (H^{T}\otimes I_n)(I_m\otimes H^{T}) + (I_n\otimes H^{T})(H^{T}\otimes I_m) = 2\,H^{T}\otimes H^{T},
\qquad (BB^{T})_{10} = 2\,H\otimes H .
\]
This is the factor of $2$ quoted in the main text; it arises because both the $\mathcal{X}$- and $\mathcal{Z}$-check columns of $B$ couple $\mathcal{Q}_0$ to $\mathcal{Q}_1$ identically.

\paragraph{Diagonalization.} Take the SVD $H = \sum_{k=1}^{r}\sigma_k\, u_k w_k^{T}$ with orthonormal left/right singular vectors $u_k\in\mathbb{R}^m$, $w_k\in\mathbb{R}^n$, so \[Hw_k = \sigma_k u_k\text{ , }H^{T}u_k = \sigma_k w_k\text{ , }H^{T}H\,w_k = \sigma_k^2 w_k\text{ , }HH^{T}u_k = \sigma_k^2 u_k.\] Extend $\{w_k\}$ to an orthonormal basis of $\mathbb{R}^n$ by kernel vectors ($\sigma=0$) and $\{u_k\}$ likewise by cokernel vectors. For an ordered pair $(i,j)$ with $\sigma_i,\sigma_j>0$, the two-dimensional subspace $\mathcal{S}_{ij} = \mathrm{span}\{\,w_i\otimes w_j\ (\in\mathcal{Q}_0),\ u_i\otimes u_j\ (\in\mathcal{Q}_1)\,\}$ is invariant under $BB^{T}$:
\[
(BB^{T})_{00}(w_i\otimes w_j) = (\sigma_i^2+\sigma_j^2)\,w_i\otimes w_j, \qquad
(BB^{T})_{01}(u_i\otimes u_j) = 2\,(H^{T}u_i)\otimes(H^{T}u_j) = 2\sigma_i\sigma_j\, w_i\otimes w_j,
\]
and symmetrically on $\mathcal{Q}_1$. Hence
\[
BB^{T}\big|_{\mathcal{S}_{ij}} = \begin{pmatrix} \sigma_i^2+\sigma_j^2 & 2\sigma_i\sigma_j \\ 2\sigma_i\sigma_j & \sigma_i^2+\sigma_j^2 \end{pmatrix},
\qquad \text{with eigenvalues } (\sigma_i+\sigma_j)^2 \text{ and } (\sigma_i-\sigma_j)^2 .
\]
Thus $B$ has singular values $\sigma_i+\sigma_j$ and $|\sigma_i-\sigma_j|$ (the \emph{product modes}). If one factor is a kernel/cokernel direction ($\sigma_j = 0$), the coupling term $2\sigma_i\sigma_j$ vanishes and $w_i\otimes w_{j_0}$ (respectively $u_i\otimes u_{j_0}$) is an eigenvector of $BB^{T}$ with eigenvalue $\sigma_i^2$, i.e.\ a singular value $\sigma_i$ of $B$ (the \emph{boundary modes}). If both factors are kernel/cokernel directions the eigenvalue is $0$ (zero modes). These subspaces are mutually orthogonal and their dimensions sum to $n^2+m^2 = \dim(\mathcal{Q}_0\oplus\mathcal{Q}_1)$, exhausting $\mathrm{sv}(B)$.

Since $\mathrm{spec}(A_{\GT}) = \{\pm\sigma : \sigma\in\mathrm{sv}(B)\}\cup\{0\}$, the Tanner adjacency spectrum is contained in
\[
\{\pm(\sigma_i\pm\sigma_j) : 1\le i,j\le r\}\ \cup\ \{\pm\sigma_k : 1\le k\le r\}\ \cup\ \{0\},
\]
being the multiset~\eqref{eq:spec-multiset}, proving Proposition~\ref{prop:spec-decomp}. The Perron value is $2\sigma_1$ (pair $(1,1)$), and when $\sigma_1$ is simple it is the unique product mode equal to $2\sigma_1$, which is used in the proof of Theorem~\ref{thm:closed-form}. 

\bibliography{p3_refs}

\end{document}